\documentclass[%
 reprint,
 amsmath,amssymb,
 aps,
]{revtex4-2}

\usepackage{graphicx}
\usepackage{dcolumn}
\usepackage{bm}
\usepackage{physics}

\graphicspath{{images/}}
\usepackage{amssymb}
\usepackage{amsmath}

\usepackage{tikz}

\usepackage{hyperref}

\usepackage{amsthm}

\usepackage{mathrsfs}
\usepackage{xr}
\newcommand{\sign}{\text{Sign}}

\theoremstyle{plain}
\newtheorem{theorem}{Theorem}[section]

\theoremstyle{definition}

\newtheorem{prop}[theorem]{Proposition}

\theoremstyle{remark}

\newtheoremstyle{mydef}
	{\topsep}   
    {\topsep}   
    {}  
    {0pt}       
    {\bfseries} 
    {.}         
    {5pt plus 1pt minus 1pt} 
    {}          

\theoremstyle{mydef}
\newtheorem{definition}[theorem]{Definition}

\newtheorem{corollary}[theorem]{Corollary}

\begin{document}

\title{A Topological Classification of Finite Chiral Structures using Complete Matchings}

\author{Maxine M. McCarthy}

\author{D. M. Whittaker}
\affiliation{%
 University of Sheffield
}%

\date{\today}

\begin{abstract} 

We present the theory and experimental demonstration of a topological
classification of finite tight binding Hamiltonians with chiral
symmetry. Using the graph-theoretic notion of complete matchings, we show
that many chiral tight binding structures can be divided into a number
of sections, each of which has independent topological phases. Hence
the overall classification is $N\mathbb{Z}_2$, corresponding to $2^N$
distinct phases, where $N$ is the number of sections with a non-trivial
$\mathbb{Z}_2$ classification. In our classification, distinct topological
phases are separated by exact closures in the energy spectrum of the
Hamiltonian, with degenerate pairs of zero energy states. We show that
that these zero energy states have an unusual localisation across
distinct regions of the structure, determined by the manner in which
the sections are connected together. We use this localisation to
provide an experimental demonstration of the validity of the
classification, through radio frequency measurements on a coaxial
cable network which maps onto a tight binding system. The structure we
investigate is a cable analogue of an ideal graphene ribbon, which
divides into four sections and has a $4\mathbb{Z}_2$ topological
classification.

\end{abstract}

\maketitle

\section{\label{sec:level1}Introduction}

Topology has become ubiquitous in modern physics, with many
lattice structures and materials shown to have non-trivial topological
attributes. A topological material is characterised by properties
which remain unchanged during adiabatic evolution, making them robust
in the presence of disorder. For instance, boundary states in
one-dimensional topological lattices may have energies which are
resistant to disorder \cite{LocalisedBoundaryModes}, while in
two-dimensions they may provide directional transport which is
protected against backscattering 
\cite{HallTopological,Haldane,NoBackScattering1,NoBackScattering}.  
In order to investigate these potentially useful
properties, it is necessary to be able to classify the different topological
phases of a given Hamiltonian.

Distinct topological phases are generally understood from two
perspectives. The first is that two phases are distinct if, and only if,
they cannot be related by (symmetry respecting) adiabatic evolution.
The second is that a non-trivial topological phase has 
interesting boundary properties. The first perspective allows us to
assign topological indices to distinct phases, and the second gives a
physical significance to a classification. Both are connected by the
famous bulk-boundary correspondence \cite{BulkBoundary1} which, since
the discovery that the quantum Hall effect is topological 
\cite{HallTopological,OGbulkBoundaryOrigin} has proven
a powerful tool for the prediction of topological boundary phenomena. 
Most approaches to classification make use of the bulk-boundary correspondence in some way.

A ubiquitous topological classification is that of the 10-fold
way \cite{10FoldWay}, in which K-theoretic methods are used to
classify each of the Altland-Zirnbauer (AZ) symmetry classes of momentum-space 
Hamiltonians \cite{AltlandZirnbauerSymmetries}. Another approach is
topological quantum chemistry \cite{TopQuantChem1,TopQuantChem2},
where the classification is obtained by looking for topologically
non-trivial states. In particular, the 
presence of states which cannot be adiabatically deformed to the
atomic limit indicates a topologically non-trivial band in the
bulk \cite{ChernInsulatorsNotLocalisable}. These methods provide a classification of stable
and fragile topological phases in a huge number of systems
\cite{TopQuantChem3}. 

A useful approach for strong disorder relies on defining a
non-commutative Brillouin zone
\cite{NonComBZmethods1,NonComBZmethods2,2014MondragenShen}. Using
non-commutative geometry, the Brillouin zone is modified to allow for
variations in each unit cell of the structure, by defining a
configuration space for the set of distinct unit cells.  Translational
invariance is no longer required, giving a disordered bulk-boundary
correspondence.

For finite structures, a spectral localiser
\cite{LocalPseudoSpectra_GeneralPaper,LocalInvariantsGaplessSystems}
can predict the presence of approximately zero energy boundary states (also by demonstrating states cannot be localised to an atomic limit),
therefore classifying a structure via the converse of the
bulk-boundary correspondence. Alternatively local realspace markers may
be used \cite{LocalTop1,LocalChernMarker} whose average gives a topological index for the entire 
structure. While this averages to a quantised value in the thermodynamic 
limit, such local markers are not exactly quantised in finite structures.

It has been shown that finite size effects may cause a 
rich sequence of topological phase transitions,
corresponding to gap closures separating \textit{topological bubbles}
\cite{FiniteInversionTopology1} in inversion symmetric structures. 
The number of approximately zero
energy boundary states in such a bubble takes an integer value
corresponding to a $\mathbb{Z}$ classification in finite, inversion-symmetric 
Hamiltonians. These states are resistant to small amounts of
disorder. Such an approach has also been extended to time-reversal
symmetric systems
\cite{TRSFiniteTopology1}.
Alternatively the finite structure may be repeated as a periodic supercell 
\cite{SupercellTreatment1,SupercellTreatment2}
allowing the use of momentum-space methods to classify the 
topological phases.

These approaches leave open a problem: how may a structure be
classified when we completely lose the bulk-boundary correspondence?
That is, we no longer have any way to define a boundary or a bulk. This
situation may occur in a small finite structure, particularly with strong disorder
and/or no underlying lattice structure (for instance, a random finite
network). 
Strongly disordered finite systems may have
a non-trivial topological classification
\cite{AnatomyOfTopAndersonTransition_As1DLocalisationStuff,LinearPaper,LocalPseudoSpectra_GeneralPaper}   
although the number of distinct topological phases is not always clear, motivating a general approach to topologically classify finite
structures which have lost the bulk-boundary correspondence. In this paper, we propose a partial solution
to this problem. Using graph-theoretic methods, we give a topological
classification for finite chiral symmetric structures. We achieve
this by considering equivalence classes of finite real space
Hamiltonians with arbitrary values for the hopping terms.

In a finite structure at sufficient disorder, a bulk may not be
possible to define, however non-adiabatic evolution may still be
defined. For a finite system, a closing in the energy spectrum is
analogous to a band gap closing in momentum space. In structures with
chiral symmetry, and thus a symmetric energy spectrum, this
corresponds to the appearance of pairs of zero energy states.
Determining the conditions on the hopping terms which lead to the presence
of these states allows us to define equivalence classes of
Hamiltonians. A similar approach to defining topological phase
boundaries has been used in
\cite{AnatomyOfTopAndersonTransition_As1DLocalisationStuff,LinearPaper,LocalPseudoSpectra_GeneralPaper}.

Using this definition of equivalence classes we find that many structures 
have a rich classification resulting from the fact that they can be divided into
a number of \textit{sections}, each of which can be assigned an independent 
topological phase. These sections are identified as corresponding to 
factors of the determinant. Such factors may be determined by 
the hopping terms that
appear in the expansion of the determinant of the Hamiltonian, and therefore play a role 
in defining the conditions for zero energy states. Any other hopping terms in the
Hamiltonian can be removed without affecting the classification. This removal process
results in the structure separating into disconnected pieces, each constituting a section.
For each section we can determine an independent 0 or $\mathbb{Z}_2$ classification, using a 
sub-Hamiltonian defined on the section.

The existence of sections depends on the connectivity of the network
representing the hopping terms in the Hamiltonian.  They can occur in structures
with some regularity, as in the graphene related examples which are considered
in this work. However, they can also appear in more randomly connected structures.
Checking 154131 non-isomorphic graphs which represent some of the connected chiral
structures with 18 sites or fewer, and using randomised searches on larger
finite chiral structures, we find that the
classification of a structure is the most rich when each site has an
average of three hopping terms. For an average of two the underlying
connectivity is generally too sparse to provide a rich classification, and at four
or greater, the structure is often too constrained by its connectivity to
allow the division into many sections.

Of the 154131 graphs classified, for all positive and real hopping terms, 14\% had no distinct topological phases, with the remaining being topologically non-trivial. 82\% had two topological phases, and 4\% had 4 to 16 distinct topological phases, with only 18 of the classified graphs having 16 phases. Including trivial and non-trivial sections 8\% of structures has no sections, 68\% had one section, and 24\% had 2 to 9 sections (trivial sections can still affect physical properties, such as localisation).

The infinite and finite classification of the sub-Hamiltonian corresponding to a section
may be different. Considering a section as a one-dimensional
supercell of an infinite lattice a transfer matrix treatment leads to one
trivial and one non-trivial topological phase, separated by a gap closure somewhere in the
Brillouin zone. Hence, with this definition, every section
has a $\mathbb{Z}_2$ classification. Furthermore, the classes AIII and BDI have a $\mathbb{Z}$ 
classification using the 10-fold way. 
In this infinite interpretation, a section having a non-trivial topological phase results 
in boundary localised topologically robust zero modes to one end of that section. 
This gives a connection to higher order topology arising through stacked chiral 
structures.

We seek to classify the actual finite 
structure, so we adopt the convention that only
unavoidable gap closures which are observable (corresponding to closure at zero momentum)
separate topological phases. This means that some sections will 
have a $\mathbb{Z}_2$ classification, while others will be topologically trivial.
The complete structure behaves as a stack of sections,
connected together in such a way that the topological phase of each section is
independent. This gives a classification as the direct sum
$\bigoplus_1^N \mathbb{Z}_2 \cong N\mathbb{Z}_2$ where $N$ is the
number of sections with a non-trivial classification, and depends on the underlying connectivity of the structure. 
For periodic materials, the
stacking is not unlike that seen in
\cite{StackingTopologicalPhaseTransitions}. Others have also shown that a rich classification can
follow from the underlying connectivity of a structure
\cite{Finite-lengthGraphenenanoribbons,SSHmodelsEngineeringTopPhases}. 

We also show that the localisation of the zero energy states that accompany a 
topological phase transition is
determined by the connections between the sections. We define a
partial ordering of the sections based on this connectivity,
such that each of the zero energy states created by making one section
topologically marginal spreads in only one direction through the
remaining sections. This provides an experimentally accessible
verification of the existence of the sections. We perform such an
experiment using a coaxial cable network, which has been shown to map
onto a tight-binding
Hamiltonian \cite{LinearPaper,DavidsArxivPaper,OtherExperimentTransmissionLineTopology,CoaxCablesPhotonicLattices}.
The structure we consider is  a zigzag
`graphene' nanoribbon consisting of four rows of sites, each of which forms
a separate section, leading to a $4\mathbb{Z}_2$ classification.

We begin by discussing the theory behind the classification in section
\ref{sec:level2} before discussing experimental results in a small
graphene structure in section \ref{sec:level3}, and concluding in
\ref{sec:level4}. Supplementary material is provided with more
mathematical and experimental details.

\section{\label{sec:level2}Classifying finite chiral structures}

\begin{figure}[hb]
\centering
\includegraphics[width=0.5\linewidth]{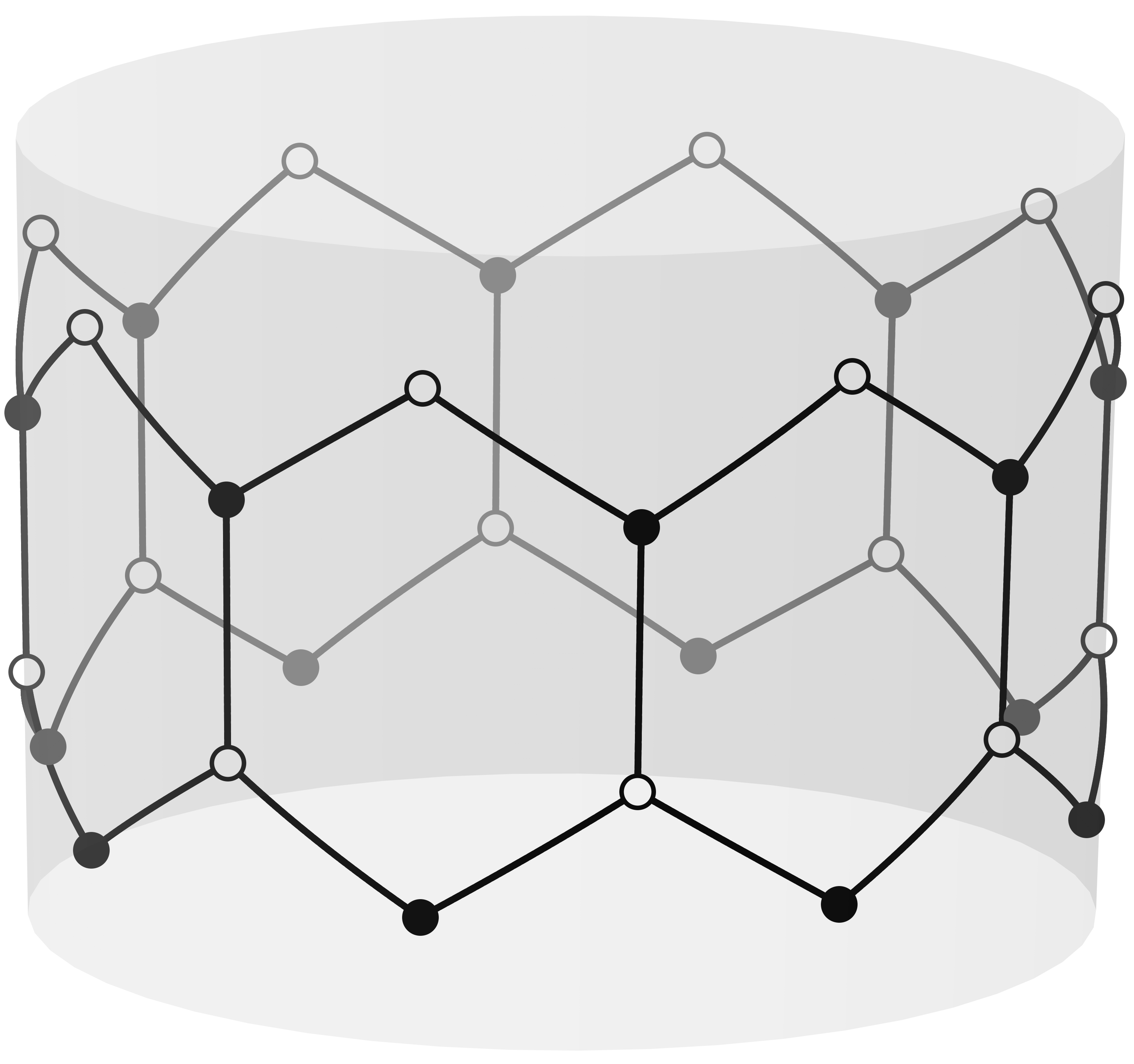}

\caption {A two row ribbon zigzag graphene, looped to form a
cylinder. The structure has chiral symmetry, with the two sublattices
indicated by the black and white sites. The topological classification
of structures like this can be found using the methods we describe: 
this one is $2\mathbb{Z}_2$.}

\label{2rowribbongrapheneexample}
\end{figure}

Chiral symmetry is one of the fundamental symmetries of the
Altland-Zirnbeaur symmetry classes \cite{AltlandZirnbauerSymmetries}.  A
chiral Hamiltonian $H$ anticommutes with a unitary $C$, ensuring
non-zero eigenvalues come in $\pm\varepsilon$ pairs. When allowing
algebraically independent hopping terms (thereby letting the structure have
arbitrary hopping disorder) a Hamiltonian with chiral symmetry necessarily has
two sublattices of sites, which we label `black' and `white' (a consequence of the Harary-Sachs theorem \cite{HarrarySachs1,HarrarySachs2}).
Each black site is only connected (through non-zero hopping terms) to white sites, and
each white site is only connected to black sites. An
example is shown in Fig.  \ref{2rowribbongrapheneexample}.

For a structure with $n_B$ black sites and $n_W$ white sites, there are
are always at least $|n_B-n_W|$ zero energy states. Since these exist
independently of the values of the hopping terms in $H$, they are described
as topologically protected states. Such states can also occur in structures
with $n_B=n_W$. If a structure has such protected states, the determinant of the Hamiltonian
is necessarily zero for all values of the hopping terms, so a classification based
on the conditions for $|H|=0$ is not possible.
In this work, we consider only structures with equal 
numbers of black and sites and no topologically protected states.

Two Hamiltonians $H_1$ and $H_2$, are considered to have a different
topological phase if they cannot be related under adiabatic evolution.
That is, there exists no continuous way to evolve between $H_1$ and
$H_2$ without a gap closing in the spectrum.
For a chiral Hamiltonian with an even number of sites,
zero energy states necessarily occur in degenerate pairs.
This degeneracy corresponds to a closed gap in the spectrum, indicating
non-adiabatic evolution. Therefore finding the conditions for a
singular Hamiltonian determines the topological phase
boundaries in a finite structure.

To explore distinct topological phases on a finite chiral
structure we consider a set of sites connected by non-zero hopping terms.
Formally this defines a graph $G$. On $G$ we define a tight binding
(TB) Hamiltonian, $H$, where each hopping term
may be continuously evolved, algebraically independently. To
ensure changes to $G$ only include continuous evolution we do not
allow new connections to be introduced, or existing connections to be broken. 
In terms of $H$, this means non-zero hopping terms must remain non-zero, 
and hopping terms which are zero cannot be modified.
Otherwise we allow complete freedom to
continuously evolve hopping terms of $H$, giving access to strongly
disordered Hamiltonians defined on $G$. We do, however, restrict 
the hopping terms to be real: any
topological phase boundary in a finite structure can be avoided, by evolving
the Hamiltonian through a path involving complex hopping terms.
Although we can allow negative values,
the requirement for non zero real hopping terms means that they cannot change sign
as they evolve.

The tuple of algebraically independent hopping terms on $G$, defines 
an affine space, $\xi$. Each Hamiltonian $H$ on $G$ defines a point
in $\xi$. Continuously evolving hopping terms in $H$ corresponds to
following a path in $\xi$. We refer to $\xi$ as the \textit{parameter
space} of $G$. Topological phase boundaries in $\xi$ correspond to
boundary free surfaces one dimension lower than $\xi$. An example of a
slice of a parameter space is given in Fig.
\ref{ParamSpaceSlice2rowribbongraphene}.

The classification problem can be completely reduced to finding
solutions to $|H|=0$, which necessarily corresponds to degenerate
zero energy states. Although degeneracies can occur at non-zero energies,
in a system with  algebraically independent
hopping terms, such a gap closure requires at least two
constraints on the hopping terms, a consequence of the Harary-Sachs theorem
\cite{HarrarySachs1,HarrarySachs2}. Such constraints are 
described by surfaces which are at least two dimensions lower
than $\xi$, and therefore cannot divide it: they are always avoidable gap closures.
This ensures topological phase
boundaries are only given by a collection of surfaces in $\xi$
corresponding to $|H|=0$.

The Hamiltonian for a chiral
structure has an antidiagonal block basis corresponding to ordering the sites by sublattice,
\begin{align}
 H=\begin{pmatrix}
   0 & Q \\
  Q^\dagger & 0
   \end{pmatrix} 
  \;,
\end{align} 
so the determinant $|H| = -|Q| |Q^{\dagger}|$. For equal numbers of black and white sites, $Q$
is a square matrix, so $|Q|=|Q^{\dagger}|$, and we can use either to determine the classification.

Every term in the expansion of $Q$ is
algebraically independent, so $|Q|$ is a multi linear form of the
hopping terms. That is $|Q|$ varies linearly with respect to each
hopping term in its expansion. 
For each hopping term in $|Q|$, solutions to $|Q|=0$
follow from solving a linear equation, constraining exactly one
hopping term. This ensures that every solution to $|Q|=0$ corresponds to a surface
in $\xi$ which is both unbounded, and one
dimension lower than that of the parameter space, splitting $\xi$ in to
distinct regions. This linear behaviour also ensures that $|Q|$
changes sign as we go over a phase boundary, meaning
$\sign\left[|Q|\right]$ is a topological invariant \cite{LocalPseudoSpectra_GeneralPaper}.
This is in contrast to 
$\sign\left[|H|\right]$, which does not change at a phase boundary.

To classify a structure, we must understand how to set $|Q|$
to zero by continuously evolving hopping terms, which may be done by
finding the irreducible factors of $|Q|$. If a factorisation of $|Q| =
\prod |q_i|$ exists, and it is possible to solve $|q_i|=0$, then the
factor $|q_i|$ defines a pair of distinct phases, where
$\sign\left[|q_i|\right]$ is a topological invariant. We say such a
factor is \textit{non-trivial}.  The classification of the structure
is then given by $\bigoplus^N \mathbb{Z}_2 \cong N\mathbb{Z}_2$ where
$N$ is the number of non-trivial factors.

The determinant $|Q|$ defines a polynomial in a unique
factorisation domain (details are given in the supplementary material section A)
ensuring the maximum number of non-trivial factors, $N$, of $|Q|$ is also
a topological invariant, so the classification is well defined.
It is only possible to change $N$  by removing or
introducing new hopping terms or sites, that is, under discontinuous
evolution of $H$. 

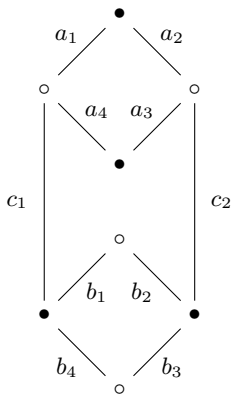
\begin{figure}[hbpt!]
\centering
\begin{tikzpicture}
\node (a) at (0,0) {$\circ$};
\node (b) at (1,1) {$\bullet$};
\node (c) at (2,0) {$\circ$};
\node (d) at (1,-1) {$\bullet$};

\node (a1) at (0.5-0.2,0.5+0.2) {$a_1$};
\node (a2) at (1.5+0.2,0.5+0.2) {$a_2$};
\node (a3) at (1.5-0.2,-0.5+0.2) {$a_3$};
\node (a4) at (0.5+0.2,-0.5+0.2) {$a_4$};

\node (e) at (0,0-3) {$\bullet$};
\node (f) at (1,1-3) {$\circ$};
\node (g) at (2,0-3) {$\bullet$};
\node (h) at (1,-1-3) {$\circ$};

\node (b1) at (0.5+0.2,0.5-0.2-3) {$b_1$};
\node (b2) at (1.5-0.2,0.5-0.2-3) {$b_2$};
\node (b3) at (1.5+0.2,-0.5-0.2-3) {$b_3$};
\node (b4) at (0.5-0.2,-0.5-0.2-3) {$b_4$};

\node (c1) at (-0.2*1.8,-1.5) {$c_1$};
\node (c2) at (2+0.2*1.8,-1.5) {$c_2$};

\draw (a) -- (b) -- (c) -- (d)--(a);

\draw (e)--(f)--(g)--(h)--(e);

\draw (a)--(e);
\draw (c)--(g);
\end{tikzpicture}
\caption{
A simple two row graphene zigzag structure with the hopping terms labelled as 
in  Eq.\eqref{SimpleExampleHamiltonian}. This structure
has two sections and four distinct topological phases, corresponding to a 
$2\mathbb{Z}_2$ classification.}
\label{Eq1Pic}
\end{figure}

To show a simple  example, we consider a 2 row zigzag graphene sructure, 
with 4 sites on each row, see Fig. \ref{Eq1Pic}. 
This is described by the matrix
\begin{equation} \label{SimpleExampleHamiltonian}
Q= \begin{pmatrix}
a_1 & a_4 & c_1 & 0 \\
a_2 & a_3 & 0 & c_2 \\
0 & 0 & b_1 & b_2 \\
0 & 0 & b_4 & b_3
\end{pmatrix}
\end{equation}
giving the determinant,
\begin{equation}
|Q| = (a_1a_3 - a_2a_4)(b_1b_3 - b_2b_4) = |q_1||q_2|.
\end{equation}
The condition $|Q|=0$ can be satisfied by making either factor, $|q_1|$ or $|q_2|$,
equal to zero. This gives four distinct topological phases corresponding to
$(\sign[|q_1|],\sign[|q_1|]) = (+1,+1),(+1,-1),(-1,+1),(-1,-1)$.
Therefore this structure has the classification $2\mathbb{Z}_2$.

It is also apparent from this simple calculation that some of the 
hopping terms -- $c_1$ and $c_2$ -- do not appear in the expansion of $|Q|$.
This means that the topological classification will be the same as for a structure
in which these hopping terms have been set to zero. Physically, removing these links
creates two completely separate lattices, the top and bottom loops in Fig.\ref{Eq1Pic}. 
We call these the two \textit{sections} of the original structure.

The Hamiltonian, and thus $Q$, for a separated structure is block diagonal, so the determinant,
$|Q|$, is simply a product of the two block determinants, as the expansion shows.
The absence of the $c$ terms from the expansion of $|Q|$ is a consequence of the 
block triangular form of Eq.(\ref{SimpleExampleHamiltonian}). However, this pattern
is only explicit if the sites are ordered correctly, which depends on finding the sections
defining the blocks in the matrix. As we prove in the supplementary material section A.1, 
the determinant is factored if and only if there exists a block triangular 
form of $Q$, so finding an appropriate ordering of sites allows us to find the number 
of sections in a structure.

A graph-theoretic approach to this problem involves enumerating the
\textit{complete matchings} of a structure.  A
complete matching consists of a set of pairings of sites, or
\textit{matchings}, which are connected (by a non-zero hopping term),
such that every site in the structure is included in exactly one pair.
Examples of such complete matchings are shown in Fig.
\ref{CoversExample}, for the four row graphene ribbon which we
investigate experimentally in section \ref{sec:level3}. A complete matching has an algebraic
interpretation: if we index separately the black and white sites such
that each each matched pair has the same index, the corresponding
hopping terms appear along the diagonal of the matrix $Q$. The terms
in the expansion of the determinant of a matrix correspond to the
product of the diagonal elements for every possible permutation of the
columns (or rows). Hence finding the set of complete matchings
for a structure gives all the terms in the expansion (to get the
signs, it is also necessary to keep track of the number of swaps
required to go between matchings).  This is the Harary-Sachs theorem
\cite{HarrarySachs1,HarrarySachs2} for weighted bipartite graphs.

A consequence of
the relationship between complete matchings and the expansion of $|Q|$
is that any connection which is not included in any matching does not appear 
in the expansion, so can be removed to reveal the sections. In the case of the structure in 
Fig. \ref{CoversExample}, all the connections between the rows can be removed, so the structure
has four sections. For each section, there are two complete matchings, so each factor $|q_i|$ of $|Q|$ 
contains two terms. Since these have opposite signs, every factor can be made to pass through
zero, so the structure has classification $4\mathbb{Z}_2$. Note that even small changes in 
connectivity can change this completely. If we add just one connection between the a white site of top row and 
a black site of the bottom row of the structure, we find that there are complete matchings which include every connection,
so there are no sections and the classification is $\mathbb{Z}_2$.

\begin{figure}[hbt!]
\centering
\includegraphics[width=\linewidth]{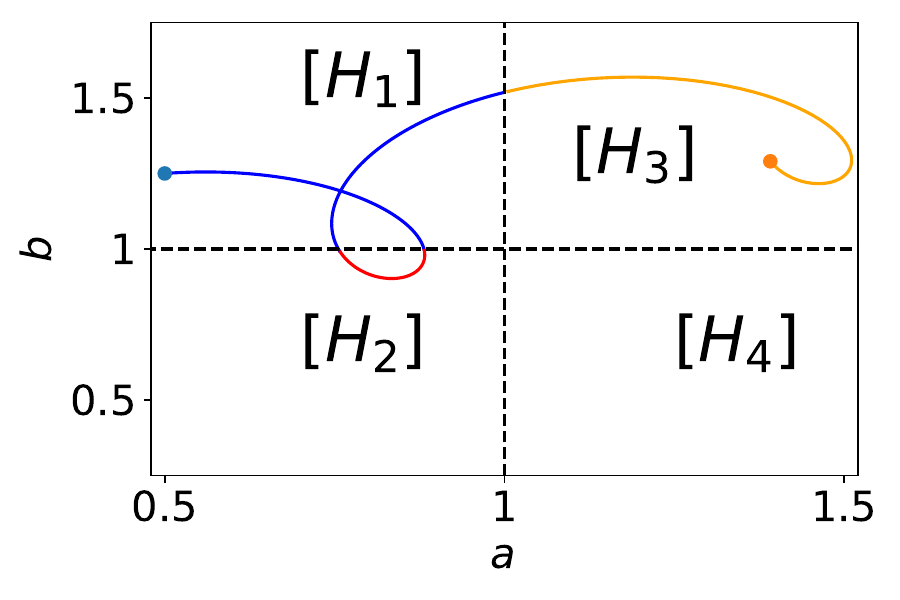}
\caption{A slice of the parameter space for 2 row ribbon zigzag graphene, which has 
the classification $2\mathbb{Z}_2$. This 
slice is defined by setting all the hopping terms equal to 1, apart from one in each
row,  
$a$ and $b$, which are allowed to evolve independently. 
$[H_x]$ denotes a set of topologically equivalent Hamiltonians, and the dashed 
lines denote phase boundaries. The path shown  undergoes 3 phase transitions between 
the phases $[H_1]$, $[H_2]$, and $[H_3]$, where each phase transition is identified
by the appearance of a pair of unavoidable degenerate zero energy states.}
\label{ParamSpaceSlice2rowribbongraphene}
\end{figure}

\begin{figure}[hbt!]
\centering
\includegraphics[width=0.28\linewidth]{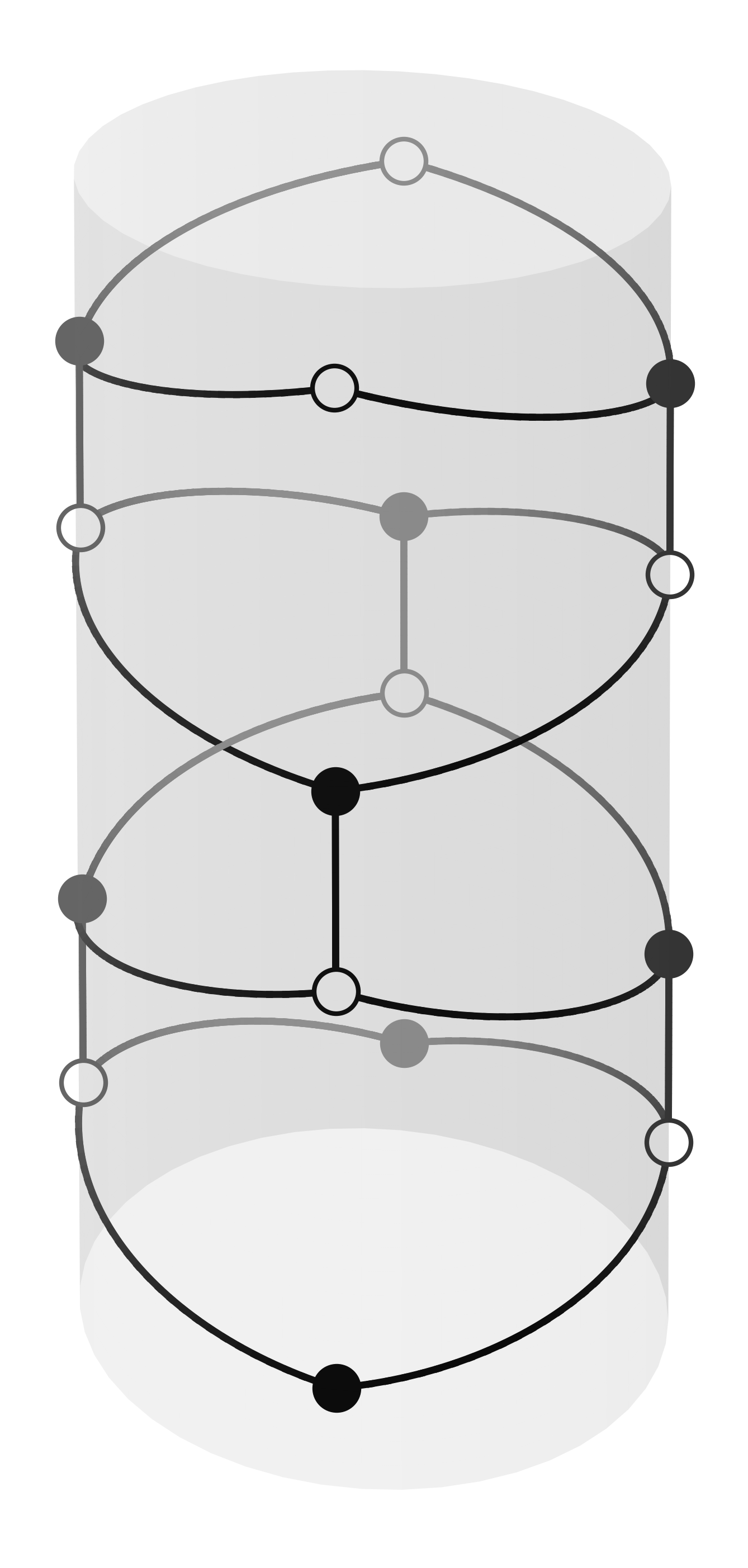} 
\hspace{1.5em}
\includegraphics[width=0.5\linewidth]{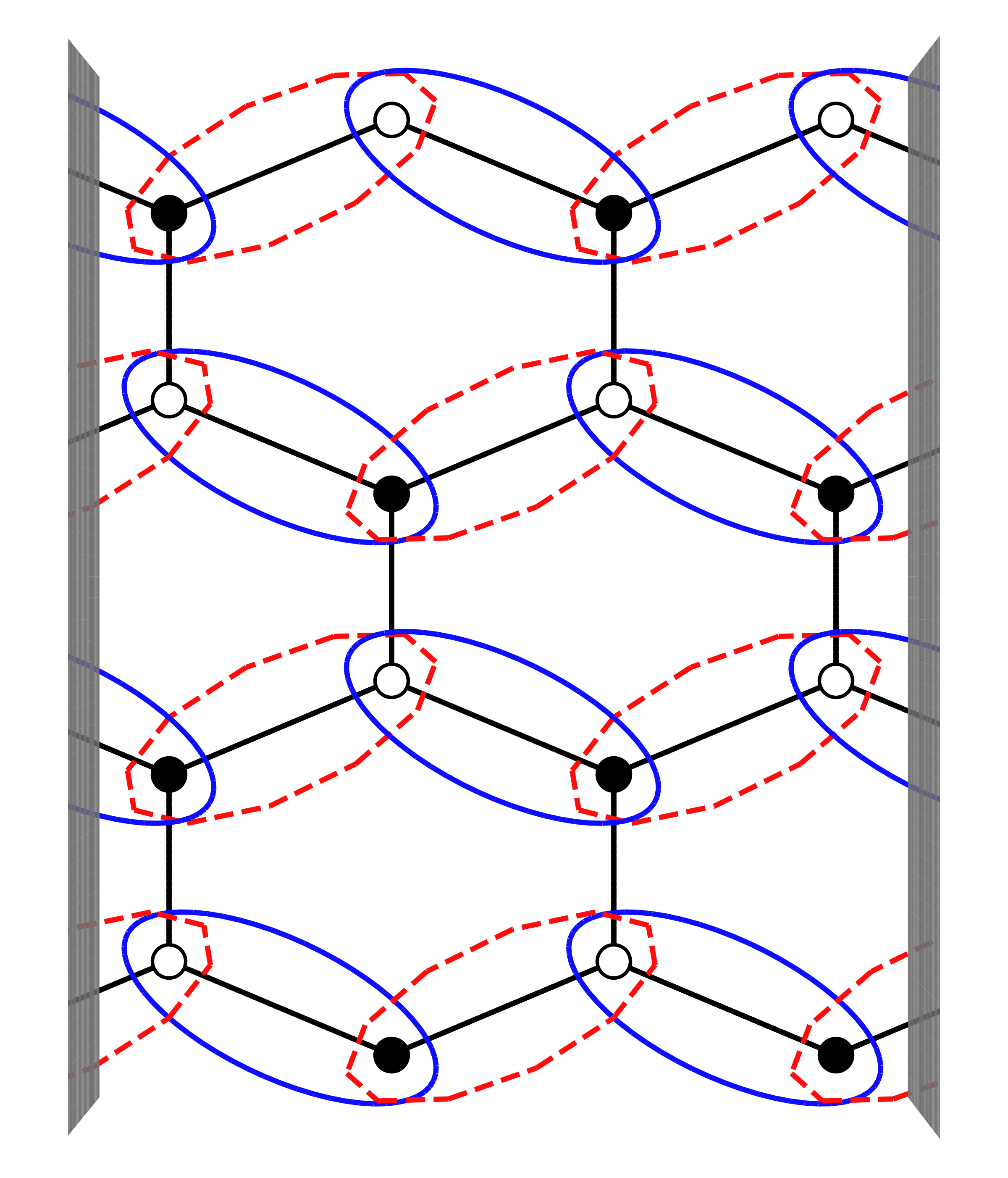}\\
(a) \hspace{9.5em} (b)

\caption{(a) The four row ribbon graphene structure looped to form a cylinder.
(b) The complete matchings of this structure, flattened for clarity. 
A complete matching consists of a solid or 
dashed set of matchings from each row. The choice for
each row is independent so there are $2^4$ complete matchings.
The connections between the rows do not appear in any of the 
matchings, so each row forms a separate section. the topological
classification for this structure is $4\mathbb{Z}_2$.
}
\label{CoversExample}
\end{figure}

The relationship between the number of sections of a structure and a block triangular form of 
$Q$ in Eq.(\ref{SimpleExampleHamiltonian}) can be made rigorous. In 
the supplementary material section A.1, we prove that the factorisation of $|Q|$ according to
sections, $|Q|=\prod_i |q_i|$, can be made if, and only if $Q$ can be written in the 
block triangular form 
\begin{equation} \label{TriangularBlockForm1}
Q = \begin{pmatrix}
q_j & \cdots & c_{i+1,j} & c_{i,j} \\
& \ddots & \vdots & \vdots \\
& & q_{i+1} & c_{i,i+1} \\
\text{\huge0} & & & q_i \\
\end{pmatrix}
\;,
\end{equation}
where the $q_i$ and $c_{ij}$ are matrix blocks. Therefore by finding an 
appropriate order of the sites the number of sections will be revealed 
in a structure.

The diagonal blocks, $q_i$, in 
Eq.(\ref{TriangularBlockForm1}) define the sections of the structure. They
include the hopping terms which connect the black and white sites within the section.
The off diagonal blocks, $c_{i,j}$, contain  the connections between the sections.
The hopping terms in the off diagonal blocks do not appear in the expansion of $|Q|$. Note that
the connecting terms only appear within superdiagonal blocks -- the corresponding 
subdiagonal is necessarily zero. Thus the connections between sections are always
between the sublattices. For instance if $Q$ maps from the black to the white 
sublattice, then only the white sites of $g_j$ may connect to the black sites of
$g_i$. We will make use of this shortly to define a partial ordering of the sections.

The relationship between the block structure of a matrix and the
topological classification of a chiral structure has allowed us to
find an algorithm to classify a random chiral network. Details of
this algorithm are given in the supplementary material section B. 
Of our classification
algorithm, the most numerically expensive part is finding a basis with
the largest number of triangular blocks. An alternative algorithm for
this part of the classification is detailed in
\cite{BlockTriangularFormSymmetricMatrices_ForAlgorithmDiscussion}
although we have not compared the complexity of our algorithm to this
alternative approach.

The factorisation and triangular form of $Q$ tells us more than just the topological
classification of the Hamiltonian. Physically the nature of
localisation at criticality (that is, at a topological phase boundary) is affected, with nullstates exactly
restricted to a particular subset of sections. We consider a 
structure represented by the graph $G$, with a section $g_i$ corresponding to the factor $q_i$.
For a topologically marginal chiral structure
each nullstate can be localised to the black or white
sub lattice. If only one section is singular, the black state may be
non-zero only on a subset $G_b$ of non-critical sections, and the
white state may be non-zero on a subset $G_w$ of non-critical
sections. The sets $G_w,\,\, G_b$ and the critical section $g_c$ are
disjoint, so that the only section with support of both the black and
white states is $g_c$. This non-trivial localisation yields an
experimental method to find the classification of a structure,
discussed in section \ref{sec:level3}.

To understand this localisation, we define a partial ordering
of sections, $(g\in G,\leq)$, which provides the upper triangular form in Eq.(\ref{TriangularBlockForm1}).
We say that $g_i< g_j$, or $g_i$ is \textit{lower} than $g_j$, 
if white sites of $g_i$ connect to black sites
of $g_j$, and $g_i> g_j$, or $g_i$ is \textit{higher} than $g_j$, 
if black
sites of $g_i$ connect to white sites of $g_j$. 
If the blocks $q_i$ and $q_j$ can be permuted in $Q$ to swap
places on the diagonal in such a way that maintains an upper
triangular block matrix $Q$, then $g_i=g_j$. Otherwise if $g_i$ and
$g_j$ are not directly connected, we look to neighbouring sections to
define the partial ordering. As the structure is connected, we may
always iterate to neighbours of neighbours  until
we have the partial ordering relationship between any two sections.

In order to demonstrate how the localisation of nullstates in a
critical section is altered by the partial ordering of sections,
consider a 4 section structure where $Q$ has the form
\begin{equation}
Q = \begin{pmatrix}
q_1 & c_{12} & c_{13} & 0 \\
0 & q_2 & 0 & c_{24} \\
0 & 0 & q_3 & c_{34} \\
0 & 0 & 0 & q_4
\end{pmatrix}
\end{equation}
where $q_i$ and $c_{i,j}$ are block matrices. The sections have the
partial ordering $g_1<g_2=g_3<g_4$. 
Suppose that section $g_3$ is marginal, so $|q_3| = 0$ and
$|q_1|,|q_2|,|q_4| \neq 0$. Then there exists a null eigenvector
$\ket{\phi}$ such that $q_3\ket{\phi} = 0$. The solution over all of
$Q$ is then given by

\begin{equation} \label{PartialOrderingDependance1}
\begin{pmatrix}
q_1 & c_{12} & c_{13} & 0 \\
0 & q_2 & 0 & c_{24} \\
0 & 0 & q_3 & c_{34} \\
0 & 0 & 0 & q_4
\end{pmatrix}
\begin{pmatrix}
-q_1^{-1}c_{13}\ket{\phi} \\
0 \\
\ket{\phi} \\
0
\end{pmatrix} = 0
\end{equation}
demonstrating the nullstate can have non-zero support on the black sites of $g_1$ and $g_3$ only. Similarly $|q_3^{\dagger}|=0$ so we have a similar solution for $Q^{\dagger}$ where
\begin{equation} \label{PartialOrderingDependance2}
\begin{pmatrix}
q_1^{\dagger} & 0 & 0 & 0 \\
c_{12}^{\dagger} & q_2^{\dagger} & 0 & 0 \\
c_{13}^{\dagger} & 0 & q_3^{\dagger} & 0 \\
0 & c_{24}^{\dagger} & c_{34}^{\dagger} & q_4^{\dagger}
\end{pmatrix}
\begin{pmatrix}
0 \\
0 \\
\ket{\psi} \\
-(q_4^{\dagger})^{-1}c_{34}^{\dagger}\ket{\psi}
\end{pmatrix} = 0
\end{equation}
and $q^{\dagger}_3\ket{\psi}=0$. Hence this nullstate has non-zero support on the white sites of
$g_3$ and $g_4$ only. The localisation of the support by sublattices
is a direct consequence of the partial ordering, in this
example $g_1<g_2=g_3<g_4$. These rules generalise in a straightforward way when there are more 
sections: the black zero energy state is localised on the marginal section and those lower
in the partial ordering, while the white state is confined to the marginal section and those
which are higher.

It should be noted that within the critical section itself, for a
randomly selected distribution of hopping terms that satisfy the
condition for a section to be critical, the zero energy states have, with
probability one, support on every site of the section. Proofs of this,
and of the general relationship between the factorisation and
localisation, are given in the supplementary material section A.1.

It is natural to ask what happens when a structure has more than one
critical section. In some instances it is possible to have more than
two zero energy states.  However, for most structures, for almost all
sets of hopping terms, there remain only two zero energy states when
there are multiple critical sections.  This is because the null states
originating from each critical section must be orthogonal to each
other, which imposes additional constraints on the hopping terms,
including those which connect the sections.  The conditions for such
higher nullity will be explored in future work.

\section{\label{sec:level3}Experimental demonstration of classification}

To demonstrate the topological classification defined in Section
\ref{sec:level3}, we have performed experiments on a coaxial cable
network which represents the 4 row ribbon graphene structure of
Fig. \ref{CoversExample}.  A network where all
coaxial cables have the same transmission time, $\tau$, maps on to a
tight binding model \cite{LinearPaper,DavidsArxivPaper} where
the sites are the junctions in the network.
The `energy', $\varepsilon$, is
given by $\varepsilon = \cos\omega \tau$ where $\omega$ is the driving
frequency. This yields a
Schr\"{o}dinger type eigenvalue equation $H\psi = \varepsilon\psi$.
Entries in $\psi$ correspond to scaled voltages at the junctions. 
Up to this scaling factor, individual hopping terms are given by
the reciprocal of the impedance of a cable connecting the sites.
Swapping between cables of different impedances  thus
allows us to traverse a structure's parameter space
$\xi$. 

Experimental measurements are made with a vector network analyser
(VNA) and takes two forms. From on-site reflectance
measurements, we determine the structure's impedance,
the real part of which is proportional to the local density of states (LDOS) at that
site \cite{LinearPaper}. We use this to demonstrate the localisation of the null states.
Transmittance measurement give the relationship between the
magnitude of a state on two sites. This gives a direct experimental determination
of the block triangular form of $Q$, Eq.(\ref{SimpleExampleHamiltonian}), and thus
the classification of the structure.

\subsubsection{Measuring the Local Density of States} \label{sec:LDOSexps}

Measuring the local density of states on every site allows us to
characterise the localisation properties of the zero energy states in a marginal structure
with multiple sections, of which only one is critical. 
Fig. \ref{LDOSexpRes} displays the LDOS measurements of the disordered
looped 4 row ribbon zigzag graphene in Fig. \ref{CoversExample}.
According to the classification of
Sec. \ref{sec:level2}, this structure has four non-trivials sections, corresponding to
the horizontal rows, so its topological
classification is $4\mathbb{Z}_2$.
Individual cables are chosen randomly from a binary distribution of
$50\Omega$ and $93\Omega$ cables (see more details in supplementary
material section C). One section is critical, the second row in the figure, while the other three are not.
The localisation of the nullstates is clear: the white state only has significant strength on 
the top two rows, while the black state appears only on the lower three. This is in agreement
with the predicted localisation, given the ordering of the rows $g_1 > g_2 >g_3 > g_4$.

\begin{figure}[hbt!]
\centering
\includegraphics[width=\linewidth]{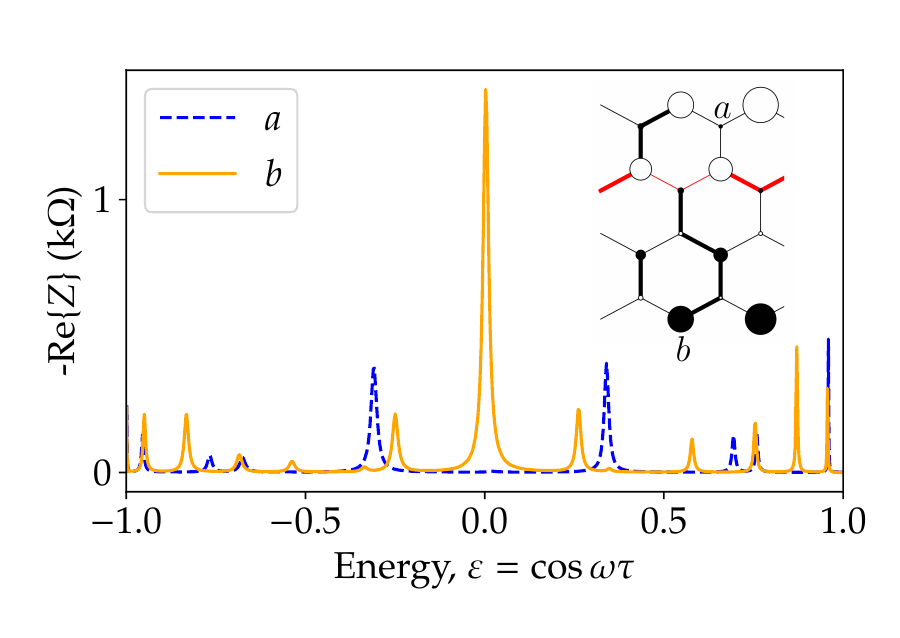}
\caption
{Experimental measurement of the impedance spectra in a topologically marginal 
4 row ribbon graphene structure. The strucure has four sections, correponding to the 
rows of the structure. Spectra are
shown for two sites on the black sublattice, labelled $a$ and $b$ in the structure diagram in the inset. 
In the plot of the structure the diameter of the circle representing each 
site is proportional the amplitude of the zero energy state at that point, which is derived 
from the strength of the corresponding impedance peak.
The widths of the lines showing the connections indicates the impedance of the corresponding 
cable, with wide lines being 93$\Omega$ and narrow lines
$50\Omega$ cables. The red lines pick out the second row which is the
critical section.  The localisation of the zero energy
eigenstates onto distinct sublattices on the upper and lower sections relative to the critical row 
agrees with our theoretical prediction. Note that the structure has
cylindrical boundary conditions as displayed in Fig.
\ref{CoversExample} (b), with the hanging connections at either side linked together.}
\label{LDOSexpRes}
\end{figure}

\subsubsection{Topological Classification through Transmittance}

Measurements of the  transmittance can be used to verify the triangular form of $Q$, 
Eq.(\ref{SimpleExampleHamiltonian}), and thus the division of the structure
into sections. We describe here an experiment where we make a cut in 
each section of the structure, by disconnecting one end of a cable, and measure 
transmittance between various cuts. We show that transmittance at 
zero energy ($\cos\omega\tau = 0$) only occurs when the measurement is within one
section: the transmittance between sections is zero.

In order to describe the transmittance of the coaxial cable network,
we use a transfer matrix formalism to relate the voltages and currents
at the output sites to those at the input sites. If we cut a site on
every section (creating an input and output site on each section), the
transfer matrix will have dimensions $2\tilde{N} \times 2\tilde{N}$, corresponding to
the voltage and current variables for each of the $\tilde{N}$ sections. Here we use $\tilde{N}$ for the total number of sections, to
distinguish from $N$, the number of topologically non-trivial sections
from earlier. 

For a chiral structure at zero energy, all
the variables can be coloured black or white, according to the two
sublattices, in such a way that the matrix consists of two $\tilde{N} \times
\tilde{N}$ diagonal blocks. At zero energy the voltage at a given site is determined entirely by the
currents flowing out of the neighbouring sites. The neighbours in a
chiral structure are on the other sublattice, so, if we assign the voltage the same colour as its site, and the currents flowing into and out of a site the 
opposite colour to the
site, the variables of the two colours are completely independent,
giving the two blocks.

To see how the form of the transfer matrix for the cut structure is related to the
sections, consider making a cut within a section, forming an input and
output site.  The cut adds a site to the structure, unbalancing the
black and white sublattices in the section, and thus creating a new
zero energy state. For instance cutting on a white site creates one
extra white site, resulting in a zero energy state that only has
support on the white sites of that section and of any higher sections
in the partial ordering.  This means that the block triangular form of
the matrix $Q$ in Eq.(\ref{SimpleExampleHamiltonian}) translates to a
triangular form of each block of the transfer matrix. 
This triangular form persists regardless of the site at which the cut is made 
in each section; it can shown to exits if, and only if, the uncut structure 
has at least $\tilde{N}$ sections.

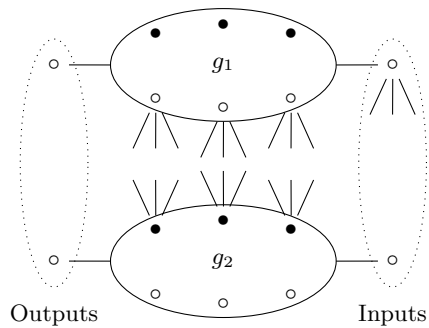
\begin{figure}[t]
\centering
\begin{tikzpicture}[scale=1.5]

\node (g2) at (0+5+5,+3.48) {$g_2$};

\node (Cw1) at (0+5+5,0.36+3.48) {$\bullet$};
\node (Cw2) at (-0.6+5+5,0.28+3.48) {$\bullet$};
\node (Cw3) at (0.6+5+5,0.28+3.48) {$\bullet$};

\node (Cb1) at (0+5+5,-0.38+3.48) {$\circ$};
\node (Cb2) at (-0.6+5+5,-0.3+3.48) {$\circ$};
\node (Cb3) at (0.6+5+5,-0.3+3.48) {$\circ$};

\draw (Cw1) -- (0+5+5,0.8+3.48);
\draw (Cw1) -- (0.2+5+5,0.8+3.48);
\draw (Cw1) -- (-0.2+5+5,0.8+3.48);

\draw (Cw2) -- (-0.6+5+5,0.72+3.48);
\draw (Cw2) -- (0.2-0.6+5+5,0.72+3.48);
\draw (Cw2) -- (-0.2-0.6+5+5,0.72+3.48);

\draw (Cw3) -- (+0.6+5+5,0.72+3.48);
\draw (Cw3) -- (0.2+0.6+5+5,0.72+3.48);
\draw (Cw3) -- (-0.2+0.6+5+5,0.72+3.48);

\node (Ccw1) at (1.5+5+5,0+3.48) {$\circ$};
\draw (Ccw1) -- (1+5+5,0+3.48);

\node (Ccw2) at (-1.5+5+5,0+3.48) {$\circ$};
\draw (Ccw2) -- (-1+5+5,0+3.48);


\draw (0+5+5,+3.48) circle [x radius=1, y radius=0.5, rotate=0];

\node (g1) at (0+5+5,+5.22) {$g_1$};

\node (Dw1) at (0+5+5,0.36+5.22) {$\bullet$};
\node (Dw2) at (-0.6+5+5,0.28+5.22) {$\bullet$};
\node (Dw3) at (0.6+5+5,0.28+5.22) {$\bullet$};

\node (Db1) at (0+5+5,-0.38+5.22) {$\circ$};
\node (Db2) at (-0.6+5+5,-0.3+5.22) {$\circ$};
\node (Db3) at (0.6+5+5,-0.3+5.22) {$\circ$};

\draw (Db1) -- (0+5+5,-0.82+5.22);
\draw (Db1) -- (0.2+5+5,-0.82+5.22);
\draw (Db1) -- (-0.2+5+5,-0.82+5.22);

\draw (Db2) -- (-0.6+5+5,-0.74+5.22);
\draw (Db2) -- (0.2-0.6+5+5,-0.74+5.22);
\draw (Db2) -- (-0.2-0.6+5+5,-0.74+5.22);

\draw (Db3) -- (+0.6+5+5,-0.74+5.22);
\draw (Db3) -- (0.2+0.6+5+5,-0.74+5.22);
\draw (Db3) -- (-0.2+0.6+5+5,-0.74+5.22);

\node (Dcw1) at (1.5+5+5,0+5.22) {$\circ$};
\draw (Dcw1) -- (1+5+5,0+5.22);

\node (Dcw2) at (-1.5+5+5,0+5.22) {$\circ$};
\draw (Dcw2) -- (-1+5+5,0+5.22);

\draw (Dcw1) -- (5+1.5+5,-0.44+5.22);
\draw (Dcw1) -- (0.2+5+1.5+5,-0.44+5.22);
\draw (Dcw1) -- (-0.2+5+1.5+5,-0.44+5.22);

\draw (0+5+5,+5.22) circle [x radius=1, y radius=0.5, rotate=0];

\draw[dotted] (11.5,+4.35) circle [x radius=0.3, y radius=1.1, rotate=0];
\draw[dotted] (8.5,+4.35) circle [x radius=0.3, y radius=1.1, rotate=0];

\node (inputs) at (11.5,3) {Inputs};
\node (inputs) at (8.5,3) {Outputs};
\end{tikzpicture}
\caption
{A structure with two sections, $g_1$ and $g_2$, with the partial
ordering of $g_1<g_2$ since white sites in $g_1$ connect to black
sites in $g_2$. The sections correspond to the  ellipses, and contain
arbitrary numbers of sites, which are not shown individually. The
hopping terms connecting the two sections are denoted schematically by bunches of three lines, but
they may join any white sites of $g_1$ to any black sites of $g_2$.
Each section has been cut  to create  input and output sites for transmittance measurements.}
\label{PotatoDiagram}
\end{figure}

As a simple example, consider a structure with just two sections, $g_1$ and $g_2$, as shown 
in Fig. \ref{PotatoDiagram}. The white sites in $g_1$ connect to the black sites in $g_2$, so in the
partial ordering,  $g_1<g_2$. Each section is cut at a white site, so the white block of the
transfer matrix relates the input and output voltages, while the black block connects the currents.
The transfer matrix is thus
\begin{equation} \label{experimentperfectshort}
\begin{pmatrix}
I^{g_1}_{\text{out}} \\
I^{g_2}_{\text{out}} \\
V^{g_1}_{\text{out}} \\
V^{g_2}_{\text{out}} \\
\end{pmatrix} = 
\begin{pmatrix}
\gamma & \delta & 0 & 0 \\
0 & \beta & 0 & 0 \\
0 & 0 & \alpha & 0 \\
0 & 0 & \nu & \mu
\end{pmatrix}
\begin{pmatrix}
I^{g_1}_{\text{in}} \\
I^{g_2}_{\text{in}} \\
V^{g_1}_{\text{in}} \\
V^{g_2}_{\text{in}} \\
\end{pmatrix},
\end{equation}
where the non-zero matrix entries represented by Greek letters are functions of the hopping 
amplitudes in the Hamiltonian, and depend on the details of the actual structure.

It is not possible to measure directly the form of such a transfer matrix
using a two-port VNA, which can only give the transmittance between
one input and one output site. The voltages and currents at the other
input and output sites cannot simply be made zero: we can either leave them open
circuit, in which case there can be a voltage on the site, but no current
flowing in or out, or they can be shorted, giving a current but no voltage. In our experiment,
the output site is always open circuit, and the choice of whether to short or leave open the 
input site is made according to the partial ordering. 

To see how this works, consider the case where the input and
output are connected to section $g_1$ and the input to $g_2$ is open circuit.
Then $I_\text{in}^{g_2}=0$ and $I_\text{out}^{g_2}=0$. From
Eq.(\ref{experimentperfectshort}), we get $I_\text{out}^{g_1}=\gamma
I_\text{in}^{g_1}$ and $V_\text{out}^{g_1}=\alpha V_\text{in}^{g_1}$,
so non-zero transmittance occurs.  However, if the input to $g_2$ is shorted, the
boundary conditions instead become $V_\text{in}^{g_2}=0$ and
$I_\text{out}^{g_2}=0$. This gives $I_\text{out}^{g_1}=0$ and
$V_\text{out}^{g_1}=\alpha V_\text{in}^{g_1}$, so no transmittance can
be seen, because this requires both the output voltage and current to
be non-zero. Hence, for our classification experiment to work, with
intra-section transmittance non-zero, we need the $g_2$ input site to be
open circuit.

Proceeding in the same way, we find that if the input port is
connected to the input site of $g_1$ and the output port to the output
site of $g_2$, we get non-zero transmittance if the input to $g_2$ is shorted,
but not if it is open circuit. Thus, to see no transmittance between
sections, we again need $g_2$ to be open circuit.

If we 
do these experiments with input port connected to the input site of $g_2$, 
the requirement is reversed: the $g_1$ input site
has to be shorted to get non-zero transmittance within the section $g_2$, and no transmittance between sections $g_1$ and $g_2$.
The reason for this difference is the partial ordering of the sections $g_1<g_2$. However, these rules only
work because we have cut the sections on the white sites. Going through the different cases,
we find that the requirement for shorting or leaving open the unused input sites depends on the sublattice
of the site and the position of the section in the 
partial ordering relative to the input site connected to the VNA. These requirements are summarised in 
Table \ref{Tab:TableOfInputStates}. When these rules are satisfied, there is non-zero transmittance
between the input and output sites when they are on the same section, but not when they are
on different sections.  This pattern is a direct consequence
of the triangular blocks in the transfer matrix, and does not occur otherwise. 
For example, in
a structure without sections but with two cuts on white sites, shorting one of
the input sites results in no transmittance for either output site.

\begin{table}[hbtp!]
\centering
\begin{tabular}{|c|c|c|c|}
\hline
Input site & Higher & Lower & Equal \\
\hline
Black & Short & Open & Open/Short \\
White & Open & Short & Open/Short \\
\hline
\end{tabular}
\caption{Conditions for shorting or leaving open circuit the inputs to the sections
where the input site is not connected to the VNA. The requirements depend on the sublattice
which the input belongs to, black or white, and its position in the partial ordering,
higher or lower than the connected section. With these choices, non-zero transmittance 
occurs only when the output site is in the same section as the input.}
\label{Tab:TableOfInputStates}
\end{table}

\begin{figure*}[hbtp!] \centering
\includegraphics[width=0.8\paperwidth]{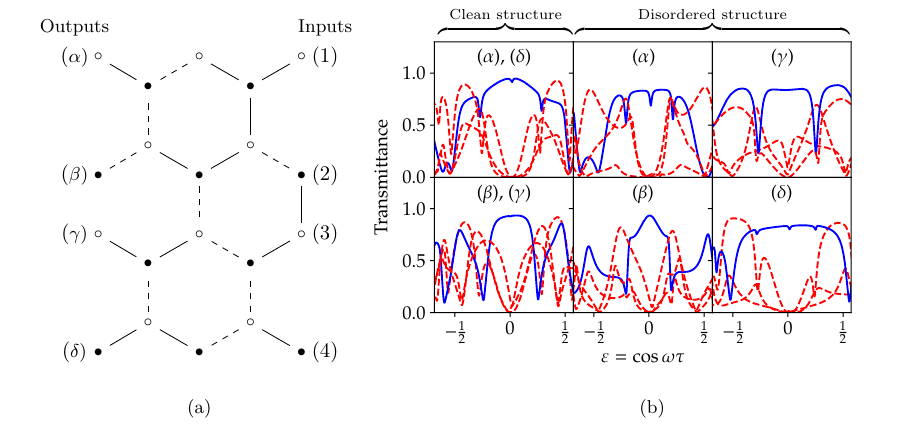}
\caption{ (a) The looped four row zigzag ribbon structure cut so it forms a flat sheet.
The looped structure is predicted to consist of four sections, 
corresponding to the four rows of sites.
Input sites for the transmittance measurements are
labelled on the right hand
side, and outputs on the left hand side. For the disordered structure,
dashed lines indicate 93$\Omega$ RG62 cables, while solid lines
indicate 50$\Omega$ RG58 cables. The clean structure had only
50$\Omega$ RG58 cables. (b) Experimental transmittance data. 
Each subplot corresponds to a different output site according to the labels in (a),
with spectra shown for all four inputs in each case. The blue spectrum
is for the case where the input and
output site are on the same section, with the dashed red spectra corresponding to 
transmittance between sections.
The transmittance at zero energy is expected to be zero if the input and output
sites are in different sections, but non-zero when they are in the same section.
The experiment thus confirms that each row is a separate section. 
}
\label{TransmittanceDataExps}
\end{figure*}

Having established these rules, we use transmittance
measurements to verify the presence of the sections in the looped four row
zigzag graphene ribbon shown in Fig. \ref{CoversExample} (which has been used as an example elsewhere in this paper). Recall that this
structure has four sections, corresponding to the four rows of sites.
Where the subscript of each row is sequentially labelled, the partial ordering of the sections 
is $g_1>g_2>g_3>g_4$. A cut is made in each
row, so the looped structure is transformed into a sheet, shown in 
Fig. \ref{TransmittanceDataExps} (a), creating an extra site on each row. Note, however, that
 the sections
we find correspond to the uncut loop rather than the sheet. In two of the sections,
the cut is on a white site, while in the other two the site is black.
We perform transmittance measurements going through the four output sites,  
and measuring a spectrum for every input sites in each case. The rules
in Table \ref{Tab:TableOfInputStates} determine whether the unconnected
input sites are left open circuit or shorted. For example, when the VNA
is connected to the input site on the second row, the first row, which is higher
in the ordering and has a white input site, is left open, while the third and fourth rows
are, respectively, shorted and open circuit. 
Using  the transfer matrix for the clean structure, where all the cables have impedance
$Z_0=50\Omega$,  this input
gives
\begin{widetext}
\begin{equation} \label{disorderedtransfermatrix} \begin{split} 
\begin{pmatrix}
V^{g_1}_{\text{out}} \\
\boldsymbol{{\color{red}I^{g_2}_{\text{out}}}} \\
V^{g_3}_{\text{out}} \\
I^{g_4}_{\text{out}} \\
I^{g_1}_{\text{out}} \\
\boldsymbol{{\color{red}V^{g_2}_{\text{out}}}} \\
I^{g_3}_{\text{out}} \\
V^{g_4}_{\text{out}} 
\end{pmatrix} = 
\begin{pmatrix}
1 & Z_0 i & -1 & - Z_0 i &  &  &  &  \\
0 & 1 & \frac{1}{Z_0} i & -1 &  &  &  &  \\
0 & 0 & 1 & Z_0 i &  &  &  &  \\
0 & 0 & 0 & 1 &  &  &  & \\
 &  &  &  & 1 & 0 & 0 & 0 \\
 &  &  &  & Z_0 i & 1 & 0 & 0 \\
 &  &  &  & -1 & \frac{1}{Z_0}i & 1 & 0 \\
 &  &  &  & -Z_0 i & -1 & Z_0 i & 1
\end{pmatrix}
\begin{pmatrix}
V^{g_1} \\
\boldsymbol{{\color{red}I^{g_2}_{\text{in}}}} \\
0 \\
0 \\
0 \\
\boldsymbol{{\color{red}V^{g_2}_{\text{in}}}} \\
I^{g_3} \\
V^{g_4}
\end{pmatrix}
=
\begin{pmatrix}
V^{g_1} + Z_0i I^{g_2}_{\text{in}} \\
\boldsymbol{{\color{red}I^{g_2}_{\text{in}}}} \\
0 \\
0 \\
0 \\
\boldsymbol{{\color{red}V^{g_2}_{\text{in}}}} \\
I^{g_3} + \frac{1}{Z_0}i V^{g_2}_{\text{in}} \\
V^{g_4} + Z_0i I^{g_3} - V^{g_2}_{\text{in}}
\end{pmatrix} 
\;.
\end{split} \end{equation}
\end{widetext}
For the disordered structure, the blocks of the matrix are still triangular, but the
expression for the elements are more complicated. The numerical value for this matrix if given in the supplementary material section C. 

If the other port of the VNA is not attached to the output site on 
the second row, it is left open circuit, that is $I_\text{out}^{g_2}=0$.
However, this forces $I_\text{in}^{g_2}$ also to be zero, so no current
enters the structure, and no transmittance occurs. However, for the output
on the same row, $I_\text{out}^{g_2}=I_\text{in}^{g_2}$ and 
$V_\text{out}^{g_2}=V_\text{in}^{g_2}$, so, in the absence of absorption, 
there is a perfect transmittance of 1. 

The transmittance experiment was carried out twice, first using a `clean' structure with 
all 50$\Omega$ (RG58)
cables, and then for a disordered structure where hopping terms were
randomly selected from a binary distribution of 50$\Omega$ (RG58) and
93$\Omega$ (RG62) cables, identical to the one in Sec. \ref{sec:LDOSexps}.
The clean structure has reflection symmetry about a line between the second and third rows,
so we need only need to use two outputs, $\alpha$ and $\beta$ (labelled in Fig.\ref{TransmittanceDataExps} (a)), but in each case we measure
transmittance for all four input sites. In the disordered structure this symmetry is broken, 
requiring measurements for every input and output site in order to complete the experiment. 

The measured transmittance spectra are shown in Fig.
\ref{TransmittanceDataExps} (b). Each panel shows data for a particular output site, with the
four spectra corresponding to the different input sites.
As predicted, in both the clean and
disordered structures, the transmittance at
$\varepsilon=\cos\omega\tau=0$ is non-zero only when the output site
is on the same section as the input (blue curves). This confirms
that the structure has $4$ sections, experimentally verifying the
$4\mathbb{Z}_2$ classification of this structure.  The actual value of the
zero-energy transmittance between the sites within the section is determined by how close 
the section is to being topologically marginal. For a marginal section, the ideal
transmittance, in the absence of losses, would be expected to have a value of unity \cite{LinearPaper}.
Although the differences are not great, it can be seen that the transmittance is highest in the 
clean structure, where all the sections are marginal, and for the second row section
of the disordered structure (output (b)), which is also marginal.

As the structure
has chiral symmetry, we expect a symmetry in the transmittance spectra
around zero energy. The slight asymmetry in the data is mainly a result of
losses in the cables, which have more effect at higher frequencies,
but there is also some chiral symmetry breaking due to imperfections in the structure.
These necessarily occur because the mapping of the coaxial cable
structure to the tight binding model requires the sections of cable
between the sites to be of uniform impedance. However, the SMA
connectors which form the junctions between the cables are $50\Omega$
components, so for the connections with $93\Omega$ cables this
uniformity is necessarily unachievable. As a result, the symmetry breaking
is generally greater for the disordered structure than the clean one, where minor 
variations in the cable lengths are the likely cause.

\section{\label{sec:level4}Conclusion}

We have described an approach to determining exact topological phases in finite chiral
structures, identifying the phase boundaries by the appearance of degenerate pairs
of zero energy states. This has been shown, in many cases, to lead to a rich topological 
classification, obtained by finding a division into sections which correspond to
irreducible factors of the determinant of the Hamiltonian. The topological 
classification is then $N\mathbb{Z}_2$, where $N$ is the number of topologically 
non-trivial sections or factors. The sections can be identified using 
the complete matchings of the underlying structure,  
relating the topological classification to the structures underlying connectivity. 
Each complete matching is related to a term in the expansion of the determinant, so a
hopping term which does not appear in any matching can be omitted without changing the determinant. 
The sections correspond to parts of the structure which become separated when these 
connections are removed. We also give, in section B of the supplementary material, a simple computational 
algorithm for finding the sections of a structure.

We have defined a partial ordering of the sections in a structure 
which gives rise to an unusual localisation of the zero energy states which are present
when one of the sections is topologically marginal. The zero energy states can be separated
so they each have support on a single sublattice. With our definition, the white
state is confined to the marginal section and those which are higher in the ordering, while the 
black state appears on the marginal section and lower sections. This localisation can be
seen as a property of finite structures which has some equivalence to the bulk boundary 
correspondence in infinite structures. 

This localisation provides a way in which the sections, and corresponding topological 
classification can be demonstrated experimentally. We have performed such experiments on
simple coaxial cable networks which map directly onto chiral tight binding models. When
we make one section marginal, we can use impedance measurements to map out the local density 
of states on each site. This provides a direct demonstration of the predicted localisation
related to the partial ordering. Even without a marginal section, we can use transmittance
measurements to show that the structure separates into the expected sections.  
We have used this method to confirm the $4\mathbb{Z}_2$ classification which our theory
predicts for a four row ribbon graphene structure.

\section*{Acknowledgements} 

We are
enormously grateful for help in the experimental work from Qingqing
Duan, Ben Kinvig, and Elena Callus, and also to Belle Darling and
Phillip Graham for huge help in finding space to run the experiments.
Many thanks also to Patrick Fowler, Barry Pickup, Qingqing Duan, Ben
Kinvig, and Elena Callus for many illuminating and insightful
discussions while completing this work.

\bibliography{classificationfinal}


\widetext
\pagebreak
\begin{center}
\textbf{\large Supplementary Materials: {A Topological Classification of Finite Chiral Structures using Complete Matchings}}
\end{center}
\setcounter{equation}{0}
\setcounter{figure}{0}
\setcounter{table}{0}
\setcounter{page}{1}
\setcounter{section}{0}
\renewcommand\thesection{\Alph{section}}
\renewcommand\thesubsection{\thesection.\arabic{subsection}}
\makeatletter
\renewcommand{\theequation}{S\arabic{equation}}
\renewcommand{\thefigure}{S\arabic{figure}}
\renewcommand{\bibnumfmt}[1]{[S#1]}
\renewcommand{\citenumfont}[1]{S#1}

\section{\label{FactThm}Factorisation theorem}

\noindent In this section we prove a theorem relating the triangular block form of $Q$ to the factorisation of $|Q|$. Consider a real or complex matrix $Q$, where $Q$ is a chiral block of a Hamiltonian
\begin{equation}
H = \begin{pmatrix}
0 & Q \\
Q^{\dagger} & 0
\end{pmatrix}.
\end{equation}
Below we show that, for algebraically independent hopping terms, then $|Q|$ is reducible if and only if there exists a way of ordering the sites such that $Q$ is upper block triangular. We refer to the resultant form of $Q$ as the maximal triangular block basis, or triangular block basis of $Q$, but it should be emphasised that this basis specifically corresponds to a permutation of $H$. We will use this in section \ref{sec:algorithm} to define an algorithm that classifies a particular structure. \\
\indent Note that often in the following arguments we will refer to a property that occurs \textit{almost always}. By this we mean that if hopping terms were selected randomly from a continuous distribution (possibly one that satisfies a certain constraint) then this property occurs with probability 1.

\indent We proceed by showing that the factorisation is well defined, before proving that a square weighted matrix has a factored determinant (for all matrix entries) if and only if it is block triangular. In proving this relationship we will show that when a section is singular, for almost all hopping terms a section has non-zero support of an eigenstate on all sites. \\ 
\indent In order to interpret determinants of a matrix as a polynomial, we consider a polynomial ring that contains them. Formally this ring is larger then the set of polynomials corresponding to determinants, but this is unimportant for our discussion. In particular we are interested in determinants of matrices, so we consider an $N\times N$ matrix,
\begin{equation}
Q_{i,j} = X_{i,j}
\end{equation}
where $X_{i,j}$ is an indeterminate over some field, or else fixed at 0, and all non-zero $X_{i,j}$ are algebraically independent. Generically we take this to be $\mathbb{R}$ or $\mathbb{C}$, but more arbitrary fields are perfectly reasonable to consider. We then take the polynomial ring $P[X_{i,j}]$ over the indeterminates $X_{i,j}$. The determinants $|Q|$ are polynomials in this ring which are linear for each hopping term in $|Q|$. To consider the matrices that may only represent tight binding models undergoing continuous evolution, then we need to restrict the domain for each indeterminate to be $\mathbb{R}^{\pm}$ or $\mathbb{C}\setminus \{0\}$ or else fixed at 0. Non-zero indeterminates are then in a semiring. \\
\indent Formally when we classify a structure we are interested in the subspace $E_0$ of $\xi$ which has no (exactly) zero energy states. We then wish to calculate the number of ways we can map a Hamiltonian to this subspace, which corresponds to calculating the zeroth homotopy group of $E_0$ under the usual topology. The zeroth homotopy group of a space counts the number of disconnected components of that space. The irreducible factorisation of $|Q|$ tells us the number of path connected components in $E_0$.
\begin{definition}
Let $X$ be the subspace of $\xi$ corresponding to a solution to $|Q|=0$. Then $E_0:=\xi\setminus X$. 
\end{definition}

\indent In order for our classification to be well defined, we need there to exist a factorisation of $|Q|$ that is irreducible and unique, that is we need the polynomial ring to be prime. This follows directly from the fact that the indeterminates are taken over a field. In other words $P$ is a unique factorisation domain. For a less abstract argument, consider the following:

\begin{prop}
If for some hopping terms $|Q|\neq 0$ then $|Q|$ has a unique irreducible factorisation in $P$.
\end{prop}

\begin{proof}
Suppose $|Q| = \prod a_i$ and $|Q| = \prod b_i$ such that $a_i, b_i$ are irreducible. Assume that for the irreducible factor $a_j$ there is no irreducible factor $b_k$ such that $a_j = m b_k$ for some non zero constant $m$, then
\begin{equation}
\prod_{i\neq j} a_i \neq \frac{\prod_{i\neq k} b_i}{m}.
\end{equation}
multiplying by $a_j = m b_k$ we see
\begin{equation}
a_j \prod_{i\neq j} a_i \neq b_k \prod_{i\neq k} b_i
\end{equation}
which is a contradiction. Hence for every $j$ there is a $k$ such that $a_j = m b_k$ for some constant $m$, showing a factorisation of $|Q|$ is unique and irreducible.
\end{proof}

\noindent Given that the irreducible factorisation of $|Q|$ is well defined, we now give a proof that $|Q| = \prod |q_i|$ if and only if there exists a block triangular structure of $Q$ with $q_i$ as diagonal blocks.

\subsection{\label{subsection:level1_factorisationproof}}

\noindent Here we aim to prove that a determinant of $Q$ is factored if and only if there exists a block triangular basis of $Q$. A corollary of one of the propositions is that a critical section almost always has non-zero support of the nullstates on every site. To prove this we first show that two sections, pairwise, may only connect on one sublattice. Then we show that for a structure with $N$ sections there must always be at least one section that may only connect to any other section on one sublattice. As this is true for any number of sections $N$ this implies that if $|Q|$ is factored, then there exists a triangular block structure of $Q$. \\
\indent More specifically, given a section $g_i$, if every first minor of $q_i$ is almost always non-zero then deleting a black and a white site from $g_i$ will result in a structure with a complete matching, because the determinant of the remaining structure is almost always non-zero. Therefore if we have two sections, $g_i$ and $g_j$ and they connect to one another on both sub lattices, then we can delete a black and a white site from both sections (each) that are connected to one another. The remaining structure has a complete matching because the associated minor is factored by a first minor of $q_i$ and a first minor of $q_j$. Consequently $g_i$ and $g_j$ can connect to one another on only one sublattice. \\
\indent Our proof relies on knowing when hopping terms are in a complete matching of some graph or not. For this we need to define a type of matching.
\begin{definition}
A \textbf{valid matching} is a matching between two sites on some graph $g$ such that the matching is part of a complete matching of $g$.
\end{definition}
\indent We then show that there cannot exist a particular type of cycle of sections (see Definition \ref{subsubsection:sectioncycle}). This ensures at least one section can only connect to other sections from one sublattice. This is true for an arbitrary number of sections, giving the triangular structure to $Q$. \\
\indent For consistency if two sections $g_i$ and $g_j$ connect from a black site in $g_i$ to a white site in $g_j$ we say $g_i$ connects to $g_j$ on the black sub lattice.

\begin{prop} \label{subsubsection:weightednuts}
Given a section $g_i$ with tight binding model $h_i = \begin{bmatrix}
0 & q_i \\
q_i^{\dagger} & 0
\end{bmatrix}$ such that $|q_i|$ is irreducible, then every first minor of $q_i$ is almost always nonzero.
\end{prop}

\begin{proof}
Given that $h_i$ is chiral, a non-zero term in $|q_i|$ corresponds to a complete matching of $g_i$. Any first order minor of $q_i$ can be accessed by deleting a white and a black site from $g_i$, with the first order minors corresponding to determinants of the $(N-1)\times (N-1)$ sub matrices of $q_i$. So if there exists a complete matching of every structure where we delete a black and a white site from $g_i$ then for almost all hopping terms the minor is non-zero. \\
\indent Suppose that we have a section with the set of complete matchings $\mathscr{C} = \{C_i\}$ and delete a black site $b_1$. Match all the remaining sites possible from the complete matching $C_i$, leaving a single unmatched white site $w_1$. If an individual site is in only one valid matching, as every complete matching must include a matching on every site, this would factorise the determinant. So every site is in at least two distinct valid matchings contained in the complete matchings $\mathscr{C}$. So $w_1$ may be matched to a black site $b_2$ with a matching in $C_j$. Removing the matching containing $b_2$ which is in the complete matching $C_i$ leaves a second white site $w_2$ unmatched. This process has changed which site is the unmatched white site, see Fig. \ref{WalkOnSectionFig} for an example. Iterating this step defines a walk over the structure. If we delete the unmatched white site, we automatically get a complete matching of the remaining structure, showing the associated minor is almost always non-zero. \\
\indent We now demonstrate any white site in a section may be unmatched by such a walk. Assume there are a set $\mathscr{W}_m$ of white sites that cannot be unmatched by iterating this walk, and a set $\bar{\mathscr{W}}_u$ of white sites that can be unmatched by iterating this walk. In $C_i$ all $\mathscr{W}_m$ white sites are matched to a black site. 
If there exists a valid matching from a black site $b_1$ to a white site in $\mathscr{W}_m$ and a valid matching from $b_1$ to a white site in $\bar{\mathscr{W}}_u$, then the site in $\mathscr{W}_m$ is unmatchable. So the black neighbours of $\mathscr{W}_m$ sites do not have a valid matching to any white site in $\bar{\mathscr{W}}_u$. 
If a white site has a valid matching to $b_1$ it is possible to unmatch, so no white site $w_m \in \mathscr{W}_m$ has a valid matching to $b_1$. This partitions the sites in to two sets: one with white sites that can be unmatched $U$, and one with white sites that cannot be unmatched $M$, with valid matchings. The valid matchings over $M$ are therefore independent of the valid matchings over $U$, giving a factorisation of $|q_i|$. By contradiction all white sites in a section are possible to unmatch via such a walk. As the choice of $b_1$ was arbitrary, this proves the existence of a complete matching of a section when deleting one white and one black site, that is, every first minor of $|q_i|$ is non-zero.
\end{proof}

\begin{figure}
\centering
\begin{tikzpicture}
\node (a) at (1,0) {\footnotesize{$\circ$}};
\node (b) at (1,-3) {\footnotesize{$\bullet$}};
\node (c) at (4,-3) {\footnotesize{$\circ$}};
\node (d) at (4,0) {\footnotesize{$\bullet$}};
\node (e) at (7,0) {\footnotesize{$\circ$}};
\node (f) at (7,-3) {\footnotesize{$\bullet$}};

\draw (a)--(b);
\draw (c)--(d);
\draw (e)--(f);

\draw (a)--(d);
\draw (d)--(e);

\draw (b)--(c);
\draw (c)--(f);

\node(g) at (4,-4) {(a)};

\draw (5.5,0) circle [x radius=1.75, y radius=0.5, rotate=0];
\draw (5.5,-3) circle [x radius=1.75, y radius=0.5, rotate=0];
\draw (1,-1.5) circle [x radius=1.75, y radius=0.5, rotate=90];

\draw[dash dot] (2.5,0) circle [x radius=1.75, y radius=0.5, rotate=0];
\draw[dash dot] (2.5,-3) circle [x radius=1.75, y radius=0.5, rotate=0];
\draw[dash dot] (7,-1.5) circle [x radius=1.75, y radius=0.5, rotate=90];

\draw[dashed] (7,-1.5) circle [x radius=1.95, y radius=0.7, rotate=90];
\draw[dashed] (4,-1.5) circle [x radius=1.95, y radius=0.7, rotate=90];
\draw[dashed] (1,-1.5) circle [x radius=1.95, y radius=0.7, rotate=90];

\node (a) at (0,-5) {\footnotesize{$\circ$}};
\node (c) at (1,-6) {\footnotesize{$\circ$}};
\node (d) at (1,-5) {\footnotesize{$\bullet$}};
\node (e) at (2,-5) {\footnotesize{$\circ$}};
\node (f) at (2,-6) {\footnotesize{$\bullet$}};

\draw (c)--(d);
\draw (e)--(f);

\draw (a)--(d);
\draw (d)--(e);

\draw (c)--(f);

\draw (1.5,-5) circle [x radius=0.7, y radius=0.15, rotate=0];
\draw (1.5,-6) circle [x radius=0.7, y radius=0.15, rotate=0];

\node (c) at (1,-8) {\footnotesize{$\circ$}};
\node (d) at (1,-7) {\footnotesize{$\bullet$}};
\node (e) at (2,-7) {\footnotesize{$\circ$}};
\node (f) at (2,-8) {\footnotesize{$\bullet$}};

\draw (c)--(d);
\draw (e)--(f);

\draw (d)--(e);

\draw (c)--(f);

\draw (1.5,-7) circle [x radius=0.7, y radius=0.15, rotate=0];
\draw (1.5,-8) circle [x radius=0.7, y radius=0.15, rotate=0];

\node (a) at (3,-5) {\footnotesize{$\circ$}};
\node (c) at (4,-6) {\footnotesize{$\circ$}};
\node (d) at (4,-5) {\footnotesize{$\bullet$}};
\node (e) at (5,-5) {\footnotesize{$\circ$}};
\node (f) at (5,-6) {\footnotesize{$\bullet$}};

\draw (c)--(d);
\draw (e)--(f);

\draw (a)--(d);
\draw (d)--(e);

\draw (c)--(f);

\draw[dash dot] (3.5,-5) circle [x radius=0.7, y radius=0.15, rotate=0];
\draw (4.5,-6) circle [x radius=0.7, y radius=0.15, rotate=0];

\node (a) at (3,-7) {\footnotesize{$\circ$}};
\node (c) at (4,-8) {\footnotesize{$\circ$}};
\node (d) at (4,-7) {\footnotesize{$\bullet$}};
\node (f) at (5,-8) {\footnotesize{$\bullet$}};

\draw (c)--(d);

\draw (a)--(d);

\draw (c)--(f);

\draw[dash dot] (3.5,-7) circle [x radius=0.7, y radius=0.15, rotate=0];
\draw (4.5,-8) circle [x radius=0.7, y radius=0.15, rotate=0];

\node (a) at (6,-5) {\footnotesize{$\circ$}};
\node (c) at (7,-6) {\footnotesize{$\circ$}};
\node (d) at (7,-5) {\footnotesize{$\bullet$}};
\node (e) at (8,-5) {\footnotesize{$\circ$}};
\node (f) at (8,-6) {\footnotesize{$\bullet$}};

\draw (c)--(d);
\draw (e)--(f);

\draw (a)--(d);
\draw (d)--(e);

\draw (c)--(f);

\draw[dash dot] (6.5,-5) circle [x radius=0.7, y radius=0.15, rotate=0];
\draw[dashed] (8,-5.5) circle [x radius=0.7, y radius=0.15, rotate=90];

\node (a) at (6,-7) {\footnotesize{$\circ$}};
\node (d) at (7,-7) {\footnotesize{$\bullet$}};
\node (e) at (8,-7) {\footnotesize{$\circ$}};
\node (f) at (8,-8) {\footnotesize{$\bullet$}};

\draw (e)--(f);

\draw (a)--(d);
\draw (d)--(e);



\draw[dash dot] (6.5,-7) circle [x radius=0.7, y radius=0.15, rotate=0];
\draw[dashed] (8,-7.5) circle [x radius=0.7, y radius=0.15, rotate=90];

\node(g) at (4,-9) {(b)};

\end{tikzpicture}
\caption{(a) 3 rung ladder graphene, which has the classification $\mathbb{Z}_2$, and all 3 complete matchings of the structure denoted with a solid line, dashed line, and dash dot line. (b) Displays the structure having deleted one of the black vertices, and matching all but one white vertex. By swapping a single matching that neighbours the unmatched vertex at a time, it is possible to leave any of the white sites unmatched. Deleting the unmatched white vertex shows that for the deleted black vertex, deleting any of the white vertices leaves a structure with a complete matching.}
\label{WalkOnSectionFig}
\end{figure}
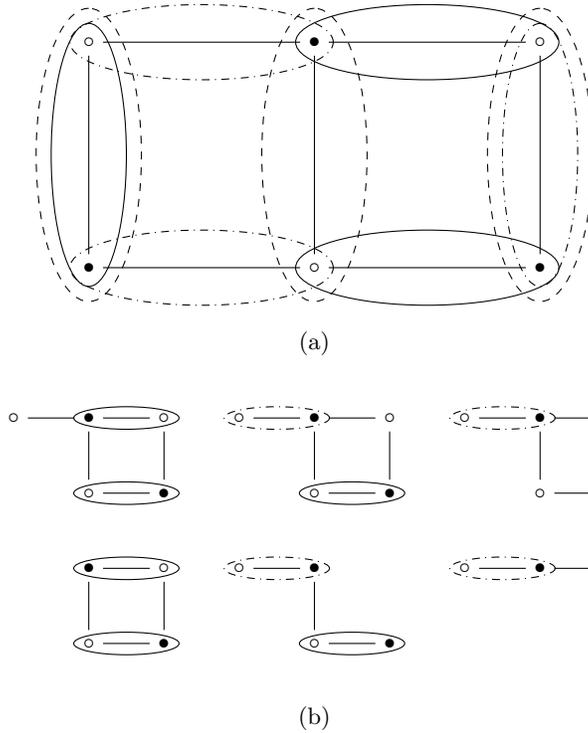

\begin{corollary}
Suppose a section is singular, then $q_i\ket{\psi} = 0$ for some non-zero $\ket{\psi}$, and for almost all hopping terms $\ket{\psi_j}$ indexed over the $j$ sites of $g_i$ is non-zero for every $j$.
\end{corollary}

\begin{proof}
By Prop. \ref{subsubsection:weightednuts} every first minor of $q_i$ is almost always non-zero. Therefore any submatrix $\bar{q}_i$ resulting from deleting a white and black site of $g_i$ is almost always non-singular. Projecting the eigenstate $q_i\ket{\psi} = 0$ on to the submatrix gives
\begin{equation}
\bar{q}_i \ket{\bar{\psi}} \neq 0
\end{equation}
for all $\bar{q}_i$. Therefore given any submatrix of $\begin{bmatrix}\bar{q}_i & \alpha\end{bmatrix}$ of $q_i$ such that
\begin{equation} \begin{split}
q_i &= \begin{bmatrix}
\bar{q}_i & \alpha \\
\hspace{1em} A
\end{bmatrix} \\
\text{and} \,\,\, \ket{\psi} &= \begin{bmatrix}
\ket{\bar{\psi}} \\
\ket{\psi_{\alpha}}
\end{bmatrix} \,\,\, \text{for} \,\,\, q_i\ket{\psi} = 0
\end{split}
\end{equation}
we see
\begin{equation}
\bar{q}_i \ket{\bar{\psi}} = -\alpha \ket{\psi_{\alpha}}.
\end{equation}
As every first minor is non-zero then for almost all hopping term values $\ket{\psi}$ has non-zero support on every site of $g_i$.
\end{proof}

\noindent We now wish to show that, given a pair of sections $g_i,g_j$ then they may only connect on one sublattice.

\begin{prop}
Given two section $g_i,g_j$ then non-zero hopping terms may only be between black sites on $g_i$ to white sites on $g_j$.
\end{prop}

\begin{proof}
Given two sections $g_i,g_j$ assume they connect to one another on both sublattices. Delete four sites, a black and a white site in $g_i$ and a white and a black site in $g_j$ that connect to the original white and black site $g_i$. By Prop. \ref{subsubsection:weightednuts} the resulting structure has a complete matching. Therefore if we put back in the four deleted sites, and match them via their associated hopping terms, then there exists a complete matching of $g_i$ and $g_j$ that is not factored by $|q_i|$ and $|q_j|$. So by contradiction $g_i$ and $g_j$ connect on only one sublattice.
\end{proof}

\noindent To show that given $N$ sections there exists a section which can only connect to any other section from one sub lattice, we first define an edge section.

\begin{definition}
An \textbf{edge section} is a section that may only connect to all other sections from one sub lattice.
\end{definition}

\noindent To prove there always exists an edge section we consider a cycle of sections.

\begin{definition} \label{subsubsection:sectioncycle}
A \textbf{section cycle} is a path through a subset of sections in $G$ such that it starts and ends on the same section following hopping terms in $g_i$ and leaves each section on a different sub lattice to the one it entered on.
\end{definition}

\noindent We now show there cannot exist a section cycle. 

\begin{prop}
A section cycle cannot exist in any structure. 
\end{prop}

\begin{proof}
Suppose every section is in a cycle, then every section connects to a distinct section on both sub lattices. For each distinct section cycle we can delete a pair of sites, one black and one white, from each section in the cycle such that the black and white sites connect to different sections in the cycle. By Prop. \ref{subsubsection:weightednuts} the resulting structure has a complete matching. Putting back in the deleted sites and hopping terms, we can now construct a complete matching that is not factored by any of the sections in that cycle. Therefore by contradiction there exist no section cycles. 
\end{proof}

\begin{corollary} \label{subsubsection:edgesection}
In a structure with $N$ sections there exists at least one edge section.
\end{corollary}

\begin{proof}
Suppose every section $g_i$ connects to at least one other section $g_j$ on the white sub lattice and one different section $g_k$ on the black sub lattice. Given such a requirement, then if no edge section exists there must be at least $N+1$ sections. Therefore by the pigeonhole principle, if every section connects on both sub lattices a section cycle exists, which is not possible by Prop. \ref{subsubsection:sectioncycle} giving a contradiction.
\end{proof}

\noindent We now wish to show that the existence of an edge section gives the triangular basis to $Q$.

\begin{theorem} \label{subsubsection:factorisationthm}
If $|Q| = \prod |q_i|$ for $|Q|,|q_i| \in P[X_{i,j}]$ then there exists a permutation of $H$ that gives a triangular block basis for $Q$.
\end{theorem}

\begin{proof}
Given a structure with $N$ sections then by corollary \ref{subsubsection:edgesection} there is at least one section, $g_1$, which only connects to all others from one sub lattice. Taking $G$ and deleting $g_1$ yields a new structure, $\bar{G}$, with $N-1$ sections. By corollary \ref{subsubsection:edgesection} this structure also has an edge section $g_2$. This may be iterated until only one section remains. This gives the partial ordering of the structure and therefore defines a triangular block basis of $Q$.
\end{proof}

\section{\label{sec:algorithm}A classification algorithm}

\noindent We present an algorithm to find the classification of a structure based on Theorem \ref{subsubsection:factorisationthm}. That is, we take some input structure, with a chiral Hamiltonian and find an ordering of the sites where $Q$ is block triangular, and every diagonal block corresponds to a section of the structure. We refer to the corresponding basis as the \textbf{maximal triangular block basis} of $Q$. Then, given the domain of each hopping term and numerical bounds on the determinant, we check if a section is trivial or not. We anticipate that this algorithm can be significantly optimised, but currently we are able to analyse up to random networks with around 50 sites. Of course, given a more consistent underlying lattice structure it often makes sense to classify a large (say many thousand sites or more) system by proving some simple results about the complete matchings of that lattice, coupled with some boundary conditions. \\
\indent There are two main parts of the algorithm: to find a triangular block basis of $Q$ with all sections on the diagonal, and then to classify each section. The former is more technically challenging. The first part uses the fact that to get a triangular block structure to and $N\times N$ size $Q$ there needs to be a $j\times (N-j)$ block of zeros, so by taking the complement of $Q$ (with all unit hopping terms) then we can find overlaps in zeros by taking dot products between column vectors, or a generalised product over several column vectors. The generalised product is defined as follows.

\begin{definition}
The generalised product of $N$ column vectors $\{c^i\}$ is given by
\begin{equation}
p = \sum_j \prod_{i} c^i_j
\end{equation}
where $j$ denotes the $j$th entry of the column vector $c^i$.
\end{definition}

Given $N-j$ columns, if $p\geq j$ then there is a $j\times (N-j)$ block of zeros in $Q$. The algorithm then follows the following general outline

\begin{enumerate}
\item Compute $Q^{\circ}$ where $Q^{\circ}_{i,j}=1$ if $Q_{i,j}=0$, otherwise $Q^{\circ}_{i,j}=0$
\item Check for $1\times (N-1)$ blocks of zeros by checking $\sum_i Q^{\circ}_{i,j}$ for every $j$. Take indices of all blocks of this size.
\item Find pairwise overlap matrix, A, of two column vectors of $Q^{\circ}_{i,j}=1$,
\begin{equation}
A_{i,j} = \sum_k Q^{\circ}_{k,i}Q^{\circ}_{k,j}
\end{equation}
\item If any $A_{i,j}\geq N-2$ there is a $2\times (N-2)$ block of zeros. Take indices of all blocks of this size.
\item For $j\geq 2$ take all columns that have overlap greater than $N-j$. If number of columns with such an overlap $\geq N-j$ set $\text{Check}=1$.
\item[] \hspace{1em} If Check==1:
\item[] \hspace{2em} For all combinations of $N-j$ columns with pairwise overlap $\geq N-j$ compute p.
\item[] \hspace{3em} If $p\geq N-j$ return indices of $j\times (N-j)$ block of zeros.
\end{enumerate}

A simple way to optimise this algorithm slightly is to project on to a sub matrix of $Q$ that excludes any individual section, whenever a relevant $j\times (N-j)$ block of zeros is found. Then the combinations of column vectors that need to be checked to find more blocks of zeros is significantly smaller. Furthermore the complexity of the search scales with the number of potential column vectors to check. So if the search for zero blocks can be done up until $j\leq \frac{N}{2}$ for the column vectors, and then switched to do a search over the remaining row vectors the algorithm may be faster. \\
\indent The second part of the algorithm is much simpler, and only applies when hopping terms are restricted to being real. By the fundamental theorem of algebra if all hopping terms are complex then a solution to $|q_i|=0$ always exists for at least a $2\times 2$ size block $q_i$ corresponding to a section. For hopping terms restricted to $\mathbb{R}^{\pm}$ we require some more computation. The first part of the algorithm computes the number of sections of $Q$, and by indexing blocks of zeros finds a maximal triangular block basis for $Q$. So we now need to see if each factor can be set to 0. This works by considering bounds for the largest value of $|q_i|$ for all $|u| \in (0.5,1]$ hopping terms, where the sign of $u$ set by the domain of that individual hopping term (i.e. if $u\in\mathbb{R}^+$ then $u>0.5$ and if $u\in\mathbb{R}^-$ then $u<-0.5$). All hopping terms are then selected as a random float in $\pm(0.5,1]$, and it is checked if this gives a singular section, if not we proceed, otherwise we reselect randomly reselect hopping terms. Suppose all hopping terms in $q_i$ are non-zero, then by Prop. \ref{subsubsection:weightednuts} every hopping term appears in the expansion of $|q_i|$. The maximum number of non-zero terms in the determinant of an $N\times N$ matrix is $N!$ and so $|q_i|\leq N!$. Therefore if a particular hopping term, $a$, is left free to vary then $|q_i| = aA + B$. If solutions to $|q_i|=0$ exist for some $a$ then for $|a| \times (0.5)^{N-1}> (N-1)!$ and $|q_i|$ has one sign, and for  $a=0$ then $|q_i| = B$ and has a different sign. This is then repeated for every hopping term until a solution is shown to exist, or it is shown to not exist for all sections. That is

\begin{enumerate}
\item Randomly select all hopping terms $u$ so that $|u|\in (0.5,1]$ and the sign corresponds to the domain of that hopping term. If $|q_i|\neq0$ proceed, otherwise randomly reselect all hopping terms until $|q_i|\neq 0$.
\item Keeping the hopping terms as in the first step set $a = \pm 2^N (N-1)!$ where the sign is specific to the domain of $a$. Calculate $\sign (|q_i|)$.
\item Keeping the hopping terms as in the first step set $a = 0$ where the sign is specific to the domain of $a$. Calculate $\sign (|q_i|)$.
\item If a change of sign is found:
\item[] \hspace{1em} Return section non-trivial
\item Else:
\item[] \hspace{1em} If all hopping terms checked:
\item[] \hspace{2em} Return section trivial
\item[] \hspace{1em} Else:
\item[] \hspace{2em} Set $a$ to its value from step 1 and repeat from step 2 for a different hopping term $b$
\end{enumerate}

\noindent This part of the algorithm can be optimised by finding a set of sites connected on a loop within the section that is itself non-trivial, then only one hopping term would need to be varied (one from this ring) for $q_i$ to be non-trivial.

\section{\label{sec:expdets}Experimental details}
\noindent Experimental data were collected with a NanoVNA V2 Plus 4 and NanoVNA V2 Plus 4 pro. Two measurements were considered: two port measurements for transmittance and a single port measurement for reflectance. To operate the VNA, the software NanoVNA-Saver was used. In making a structure, cables were taken from a binary distribution of 50$\Omega$ RG58 SMA cables, and 93$\Omega$ RG62 SMA cables. SMA connecters are generically available at only 50$\Omega$ impedances and so for the clean structure all cables were chosen to be 50$\Omega$. The exact cables used in the disordered structures are displayed in Fig. \ref{ExperimentalStructureForTransmittance}. \\
\indent Data were collected between $1-240$MHz with $\approx 114 MHz$ being the frequency at which $\varepsilon = \cos\omega\tau = 0$. \\ 
\indent For transmittance data, the four input states described in Table \ref{TableOfInputs} were measured in both the clean and disordered structure, and transmittance data taken in separate experiments on each of the four output sites.

\begin{table}[hbpt!]
\centering
\begin{tabular}{|c||c|c|c|c|}
\hline
Input site & Input 1 & Input 2 & Input 3 & Input 4 \\
\hline
$g_1$ & $(I,V)$ & Open & Open & Open \\
$g_2$ & Open & $(I,V)$ & Short & Short \\
$g_3$ & Short & Short & $(I,V)$ & Open \\
$g_4$ & Open & Open & Open & $(I,V)$ \\
\hline
\end{tabular}
\caption{The input states for the classification experiments in the graphene CCN. An entry of $(I,V)$ denotes that this is the site where the input port of the VNA is attached to the CCN.}
\label{TableOfInputs}
\end{table}

\noindent At zero energy, the transfer matrix for the disordered 4 row ribbon graphene is given by the following

\begin{equation} \label{transfermatrix}
\begin{pmatrix}
\frac{50}{93} & 100i & -\frac{143}{93} & - \frac{2500}{804357} i &  &  &  &  \\
0 & 1 & \frac{143}{2500} i & -\frac{8649}{2500} &  &  &  &  \\
0 & 0 & \frac{93}{50} & \frac{557450}{8649}i &  &  &  &  \\
0 & 0 & 0 & \frac{50}{93} &  &  &  & \\
 &  &  &  & \frac{93}{50} & 0 & 0 & 0 \\
 &  &  &  & \frac{557450}{8649}i & \frac{50}{93} & 0 & 0 \\
 &  &  &  & -\frac{1161857}{804357} & \frac{143}{8649}i & 1 & 0 \\
 &  &  &  & -\frac{125000}{8649}i & -\frac{50}{93} & 186i & \frac{93}{50}
\end{pmatrix}
\end{equation}
where numerical values are a consequence of the $50$ $\Omega$ and $93$ $\Omega$ coaxial cables. By applying the four input states of Table \ref{TableOfInputs} to equation \eqref{transfermatrix} shows the transmittance is only non-zero for sites in the same section.

\begin{figure}
\centering
\begin{tikzpicture}
\node (1) at (2*4,2*0.25) {{\footnotesize{$\circ$}}};
\node (2) at (2*4,-2*0.25) {{\footnotesize{$\bullet$}}};

\node (3) at (2*4.866,2*0.25) {\footnotesize{$\circ$}};
\node (4) at (2*4.866,-2*0.25) {\footnotesize{$\bullet$}};

\node (5) at (2*4.433,-2*0.5) {{\footnotesize{$\circ$}}};
\node (6) at (2*4.433,2*0.5) {{\footnotesize{$\bullet$}}};

\node (7) at (2*3.567,-2*0.5) {{\footnotesize{$\circ$}}};
\node (8) at (2*3.567,2*0.5) {{\footnotesize{$\bullet$}}};

\node (9) at (2*3.144,-2*0.25) {\footnotesize{$\bullet$}};
\node (10) at (2*3.144,2*0.25) {\footnotesize{$\circ$}};

\draw[dashed] (1) -- (2);
\draw (3)-- (4);

\draw (5)--(4);
\draw[dashed] (5)--(2);
\draw (6)--(1);
\draw[dashed] (6)--(3);
\draw (7)--(2);
\draw (7)--(9);
\draw (8)--(1);
\draw[dashed] (8)--(10);

\node (11) at (2*4.433,2*1) {{\footnotesize{$\circ$}}};
\node (12) at (2*4.866,2*1.25) {\footnotesize{$\bullet$}};
\node (13) at (2*4,2*1.25) {{\footnotesize{$\bullet$}}};
\node (14) at (2*3.567,2*1) {{\footnotesize{$\circ$}}};
\node (15) at (2*3.144,2*1.25) {\footnotesize{$\bullet$}};

\draw (11) -- (12);
\draw (11)-- (13);
\draw[dashed] (13)--(14);
\draw (14)--(15);

\draw (6)--(11);
\draw[dashed] (8)--(14);

\node (16) at (2*4.433,-2*1) {{\footnotesize{$\bullet$}}};
\node (17) at (2*4.866,-2*1.25) {\footnotesize{$\circ$}};
\node (18) at (2*4,-2*1.25) {{\footnotesize{$\circ$}}};
\node (19) at (2*3.567,-2*1) {{\footnotesize{$\bullet$}}};
\node (20) at (2*3.144,-2*1.25) {\footnotesize{$\circ$}};

\draw (16) -- (17);
\draw[dashed] (16)-- (18);
\draw (18)--(19);
\draw (19)--(20);

\draw[dashed] (19)--(7);
\draw[dashed] (16)--(5);

\node (0) at (2*4.866,3) {Inputs};
\node (00) at (2*3.144,3) {Outputs};

\draw[dash dot] (2*4.866,0) circle [x radius=0.3, y radius=2.7, rotate=0];
\draw[dash dot] (2*3.144,0) circle [x radius=0.3, y radius=2.7, rotate=0];

\draw[dotted] (2*3.144+1.722,2.35) circle [x radius=2, y radius=0.6, rotate=0];
\draw[dotted] (2*3.144+1.722,-2.35) circle [x radius=2, y radius=0.6, rotate=0];

\draw[dotted] (2*3.144+1.722,-0.65) circle [x radius=2, y radius=0.6, rotate=0];
\draw[dotted] (2*3.144+1.722,0.65) circle [x radius=2, y radius=0.6, rotate=0];

\end{tikzpicture}
\caption{The disordered structure the transmittance experiment was performed on, as discussed in section III 2 of the paper, with sites in each section already cut for measurements. The dashed lines denote the 93$\Omega$ RG62 cables, and the solid lines denote 50$\Omega$ RG58 cables. The same structure with all 50$\Omega$ RG58 cables is the clean structure used for the transmittance experiments. The dash-doted ellipses denote the input and output sites. The same structure with uncut sites was used for the localisation experiment in section III 1 of the paper. This structure has 4 sections and a 4$\mathbb{Z}_2$ classification.}
\label{ExperimentalStructureForTransmittance}
\end{figure}

\end{document}